\documentclass[11pt,a4paper]{article}%
\usepackage{amsfonts}
\usepackage{amsmath}
\usepackage{amssymb}
\usepackage{graphicx}
\usepackage{geometry}%
\providecommand{\U}[1]{\protect\rule{.1in}{.1in}}
\newtheorem{theorem}{Theorem}

\newtheorem{proposition}{Proposition}
\newtheorem{remark}{Remark}

\newenvironment{proof}[1][Proof]{\textbf{#1.} }{\  \rule{0.5em}{0.5em}}
\makeatletter
\def \@removefromreset#1#2{\let \@tempb \@elt
\def \@tempa#1{@&#1}\expandafter \let \csname @*#1*\endcsname \@tempa
\def \@elt##1{\expandafter \ifx \csname @*##1*\endcsname \@tempa \else
\noexpand \@elt{##1}\fi}     \expandafter \edef \csname cl@#2\endcsname{\csname cl@#2\endcsname}     \let \@elt \@tempb
\expandafter \let \csname @*#1*\endcsname \@undefined}

\@removefromreset{equation}{section}

\@removefromreset{theorem}{section}
\makeatother
\begin{document}

\title{Conclusive discrimination by $N$ sequential receivers between $r\geq2$
arbitrary quantum states}
\author{Elena R. Loubenets$^{1,2}$ and Min Namkung$^{1,3}$\\$^{1}$National Research University Higher School of Economics, \\Moscow 101000, Russia\\$^{2}$Steklov Mathematical Institute of Russian Academy of Sciences, \\Moscow 119991, Russia \\$^{3}$Department of Applied Mathematics and Institute of Natural Sciences, \\Kyung Hee University, Yongin, 17104, Republic of Korea}
\maketitle

\begin{abstract}
In the present article, we develop a general framework for the description of
discrimination between $r\geq2$ quantum states by $N\geq1$ sequential
receivers in the case where each receiver obtains a conclusive result. This
type of discrimination constitutes an $N$-sequential extension of the
minimum-error discrimination\ by one receiver. The developed general
framework, which is valid for a conclusive discrimination between any number
$r\geq2$ of arbitrary quantum states, pure or mixed, of an arbitrary dimension
and any number $N\geq1$ of sequential receivers, is based on the notion of a
quantum state instrument and this allows us to derive the new important
general results. We, in particular, find a general condition on $r\geq2$
quantum states, under which, within the strategy where all types of receivers'
quantum measurements are allowed, the optimal success probability is equal to
that of the first receiver for any number $N\geq2$ of further sequential
receivers and specify the corresponding optimal protocol. We show that, for
two arbitrary quantum states, this sufficient condition is always fulfilled so
that, within the above strategy, the optimal success probability of an
$N$-sequential conclusive discrimination between two arbitrary quantum states
is given by the Helstrom bound for any number of sequential receivers, and
specify the optimal protocols where this bound is attained. Each of these
optimal protocols is general in the sense that it is true for any two quantum
states, pure or mixed, and of an arbitrary dimension. Furthermore, we extend
our general framework to include an $N$-sequential conclusive discrimination
between $r\geq2$ arbitrary quantum states under a noisy communication. As an
example, we analyze analytically and numerically a two-sequential conclusive
discrimination between two qubit states via depolarizing quantum channels. The
derived new general results are important both from the theoretical point of
view and for the development of a successful multipartite quantum
communication via noisy quantum channels.

\end{abstract}

\section{Introduction}

Constructing quantum measurements for distinguishing between quantum states
with {an} optimal value of a figure of merit is one of {the key problems of
quantum information processing,} which is {referred to as a quantum state
discrimination \cite{c.w.helstrom,a.s.holevo1}.} {For} {a }quantum state
discrimination, various {strategies} {based on} the constraints imposed by the
quantum measurement theory can be {proposed}.

In 2013, J. A. Bergou {\cite{j.a.bergou}} firstly devised the sequential
unambiguous state discrimination {where} a receiver discriminates {a} sender's
quantum state and a further receiver discriminates posterior states from the
{previous} receiver. However, in \cite{j.a.bergou}, every receiver performs
{the} {so called unambiguous state discrimination where }each receiver's
conclusive result is always confident, even though {this} receiver can obtain
{the} inconclusive result with a non-zero probability, and until now there
have been numerous theoretical and experimental developments
\cite{c.-q.pang,j.-h.zhang,m.namkung,m.namkung2,m.hillery,m.namkung3,m.a.solis-prosser,m.namkung4}
on only this type of sequential state discrimination.

Unfortunately, the unambiguous state discrimination can be performed if only
(i) in the case of pure initial states, they are linearly independent; and
(ii) in the case of mixed initial states, a support space spanned by
eigenvectors of nonzero eigenvalue of one mixed state is not identical to that
of any other mixed state \cite{t.rudolph}. Therefore, another strategy for a
sequential state discrimination should be devised.

Recently, another\textbf{ }scenario\textbf{ }for a sequential state
discrimination has been proposed \cite{d.fields} where each receiver's quantum
measurement is performed on a quantum system in the conditional posterior
state after a measurement of the previous receiver and outputs only a
conclusive result. The latter means that receivers' quantum measurements are
designed to remove a possibility for each receiver to obtain a result
\textquotedblleft I don't know which quantum state was
prepared\textquotedblright. In view of this, we call this discrimination
scenario as a sequential conclusive quantum state discrimination. However, the
analytical and numerical results derived in \cite{d.fields} refer only to the
$N$-sequential conclusive discrimination between two pure qubit states under
receivers' generalized quantum measurements with some specific conditional
posterior states\footnote{For the constraint used in \cite{d.fields} on
receivers' quantum measurements, see section 2 of that paper.}.

In the present article, we develop a general framework for description of an
$N$-sequential conclusive discrimination between any number $r\geq2$ arbitrary
quantum states, pure or mixed, under any receivers' generalized quantum
measurements. This type of quantum state discrimination constitutes an
$N$-sequential extension of the minimum-error discrimination by one receiver
and the developed general framework incorporates the description of
$N$-sequential discrimination considered in \cite{d.fields} only as a
particular case.

Mathematically the new general framework for an $N$-sequential conclusive
state discrimination is based on the notion of a quantum state instrument
describing a consecutive measurement of $N\geq1$ receivers and this allows us
to derive three mutually equivalent representations for the success
probability via correspondingly: (i) a quantum state instrument under a
consecutive measurement by $N$ receivers; (ii) the POV measures of all $N$
receivers and conditional posterior states after each sequential measurement
and (iii) the product of the success probabilities of all $N$ receivers and to
present for the optimal success probability a general upper bound expressed
explicitly in terms of $r\geq2$ quantum states and their a priori probabilities.

We find a new general condition on $r\geq2$ quantum states sufficient for the
optimal success probability to be equal to the optimal success probability of
the first receiver and specify for this case the optimal protocol. We show
that, for the $N$-sequential conclusive discrimination between two $(r=2)$
arbitrary quantum states, this general condition is always fulfilled,
therefore, in this case, the optimal success probability is given by the
Helstrom bound for any number $N\geq1$ of sequential receivers and specify the
optimal protocols where this bound is attained. Each of these optimal
protocols is general in the sense that it is true for the $N$-sequential
conclusive sequential discrimination between any two quantum states, pure or
mixed, and of an arbitrary dimension.

Furthermore, we extend our general framework to include an $N$-sequential
conclusive state discrimination under communication via noisy quantum
channels. As an example, we apply the developed formalism for the analysis of
the two-sequential conclusive discrimination between two qubit states in the
presence of a depolarizing noise.

The developed general mathematical framework is true for any number $N$ of
sequential receivers, any number $r\geq2$ of arbitrary quantum states, pure or
mixed, any type of receivers' quantum measurements and for arbitrary noisy
quantum communication channels is important both theoretically and for the
development of a successful practical multipartite quantum communication via
noisy channels.

The new general results derived in the present article are important both from
the theoretical point of view and for the development of a successful
multipartite quantum communication via noisy channels.

The developed general framework is true for any number $N$ of sequential
receivers, any number of arbitrary quantum states, pure or mixed, to be
discriminated, and all types of receivers' quantum measurements and arbitrary
noisy quantum communication channels.

The present article is organized as follows.

In Section 2, we present the main issues of the quantum measurement theory,
specify the notions of a quantum state instrument and its statistical realizations.

In Section 3, we develop a new general framework for the description of an
$N$-sequential conclusive discrimination between $r\geq2$ arbitrary quantum
states prepared with any a priori probabilities.

In Section 4, we specify the optimal success probability under a definite
strategy for an $N$-sequential conclusive state discrimination and establish a
new upper bound on this probability. For the strategy where all possible
receivers' quantum measurements are allowed, we find a new general condition
on $r\geq2$ quantum states sufficient for the optimal success probability of
the $N$-sequential conclusive state discrimination to be equal to the optimal
success probability of the first receiver for any number $N\geq2$ of
receivers. We show that, in case $r=2$, this general condition is fulfilled
for any two quantum states, possibly infinite dimensional, therefore, within
the strategy where all possible receivers' quantum measurements are allowed,
the optimal success probability for the $N$-sequential conclusive state
discrimination between two arbitrary quantum states is given by the Helstrom
bound for any number $N\geq2$ of sequential receivers and specify two possible
optimal protocols where this bound is attained and which are general in the
sense that they are true for the discrimination between any two quantum
states, pure or mixed and of an arbitrary dimension.

In Section 5, we develop a general framework for an $N$-sequential conclusive
discrimination between $r\geq2$ arbitrary quantum states via noisy
communication channels.

In Section 6, we analyze analytically and numerically a two-sequential
conclusive discrimination between two qubit states via depolarizing quantum channels.

In Section 7, we summarize the main results.

\section{Preliminaries}

For the description in Section 3, 4 of a general scenario for an
$N$-sequential conclusive state discrimination, we shortly recall the
main\ notions of the quantum measurement theory
\cite{a.s.holevo1,e.b.davies,p.busch}. Let a measurement with outcomes
$\omega$ in a \emph{finite} set $\Omega$ be performed on a quantum system
described in terms of a complex Hilbert space $\mathcal{H}$. Denote by
$\mathcal{L}\mathfrak{(}\mathcal{H})$ the vector space of all bounded linear
operators on $\mathcal{H}$ and by $\mathcal{T}\mathfrak{(}\mathcal{H}%
)\subseteq\mathcal{L}\mathfrak{(}\mathcal{H})$ -- the vector space of all
trace class operators on $\mathcal{H}.$

The \emph{complete description }of every quantum measurement, projective or
generalized, includes the knowledge of both -- statistics of observed outcomes
and a family of posterior states, each conditioned by an outcome observed
under a single trial of this measurement. In mathematical terms, the complete
description of a quantum measurement is specified
\cite{a.s.holevo1,e.b.davies,p.busch,a.s.holevo2} by the notion of a state
instrument $\mathcal{M=\{M(\omega)},\omega\in\Omega\mathcal{\}}$ with values
$\mathcal{M(}\omega)$ that are completely positive bounded linear maps
$\mathcal{M(}\omega)[\cdot]:\mathcal{T}\mathfrak{(}\mathcal{H})\rightarrow
\mathcal{T}\mathfrak{(}\mathcal{H}),$ satisfying the relation
\begin{equation}
\sum_{\omega\in\Omega}\mathrm{tr\{}\mathcal{M}\mathcal{(}\omega
)[T]\}=\mathrm{tr\{}\mathcal{M}\mathcal{(}\Omega)[T]\}=\mathrm{tr}%
\{T\},\mathit{\ \ \ }T\in\mathcal{T}\mathfrak{(}\mathcal{H}). \label{1.1}%
\end{equation}
To each state instrument $\mathcal{M}$, there corresponds the unique
observable instrument $\mathcal{N=\{N(}\omega),\omega\in\Omega\}$ with values
$\mathcal{N(}\omega)$ that are \emph{normal} completely positive bounded
linear maps $\mathcal{L}\mathfrak{(}\mathcal{H})\rightarrow\mathcal{L}%
\mathfrak{(}\mathcal{H})$ defined to a state instrument $\mathcal{M}$ via the
duality relation
\begin{align}
&  \mathrm{tr}\left\{  \mathcal{M(}\omega)[T]Y\right\}  =\mathrm{tr}\left\{
T\mathcal{N}(\omega)[Y]\right\}  ,\mathit{\ \ }\omega\in\Omega,\label{2.1}\\
&  T \in\mathcal{T}\mathfrak{(}\mathcal{H}),\mathit{\ }Y\in\mathcal{L}%
\mathfrak{(}\mathcal{H}),\nonumber
\end{align}
and the values of the POV measure $\mathrm{M}=\mathcal{\{}\mathrm{M}%
\mathcal{(\omega)},\omega\in\Omega\mathcal{\}},$ $\sum_{\omega\in\Omega
}\mathrm{M}_{\mathcal{M}}(\omega)=\mathbb{I}_{\mathcal{H}}$, describing the
statistics of this quantum measurement are given by%
\begin{equation}
\mathrm{M}(\omega):=\mathcal{N}(\omega)[\mathbb{I}_{\mathcal{H}}]. \label{4}%
\end{equation}

Given a state instrument $\mathcal{M}$ of a quantum measurement and a quantum
system state $\rho$ (density operator) on $\mathcal{H}$ before this
measurement, the probability to observe under this measurement an outcome
$\omega$ in a subset $F\subseteq\Omega$ has the form%
\begin{equation}
\mu(F\text{ }|\rho)=\mathrm{tr}\left\{  \mathcal{M(}F)[\rho]\right\}
=\mathrm{tr}\left\{  \rho\mathcal{N}(F)[\mathbb{I}_{\mathcal{H}}]\right\}
=\mathrm{tr}\{\rho\mathrm{M}(F)\}. \label{4_1}%
\end{equation}

After a single measurement trial, where an outcome $\omega\in\Omega$ is
observed, the state of a quantum system is given by the relation
\cite{a.s.holevo2,e.r.loubenets,o.e.barndoff-nielson}
\begin{equation}
\rho_{out}(\omega|\rho):=\frac{\mathcal{M(}\omega)[\rho]}{\mu(\omega|\rho)}
\label{5}%
\end{equation}
and is called a conditional posterior state. The unconditional posterior state
is defined by $\rho_{out}(\rho):=\rho_{out}(\Omega|\rho)=\mathcal{M(}%
\Omega)[\rho].$

Every state instrument $\mathcal{M(}\cdot)[\cdot]$ admits the
Stinespring-Kraus representation\footnote{For details, see
\cite{a.s.holevo2,e.r.loubenets,o.e.barndoff-nielson,shirokov1, shirokov2} and
references therein.}%
\begin{equation}
\mathcal{M(}\omega)[T]=\sum_{l}K_{l}(\omega)TK_{l}^{\dagger}(\omega
),\mathit{\ \ }T\in\mathcal{T}\mathfrak{(}\mathcal{H}),\mathit{\ \ }\omega
\in\Omega, \label{7}%
\end{equation}
which may be not unique. Here, $K_{l}(\omega)\in\mathcal{L}\mathfrak{(}%
\mathcal{H}),$ $\omega\in\Omega,$ $l\in\{1,...,L\}$, are bounded linear
operators with the operator norms $\left\Vert K_{l}(\omega)\right\Vert \leq1$,
called the Kraus operators \cite{Kraus} and satisfying the relation%
\begin{equation}
\sum_{\omega,l}K_{l}^{\dagger}(\omega)K_{l}(\omega)=\mathbb{I}_{\mathcal{H}}.
\label{8}%
\end{equation}
From (\ref{2.1}) and (\ref{8}) it follows that, in terms of Kraus operators,
for each $\omega\in\Omega,$%
\begin{align}
&  \mathrm{M}(\omega)=\sum_{l}K_{l}(\omega)^{\dagger}K_{l}(\omega),\label{9}\\
&  \mu(\omega|\rho)=\sum_{l}\mathrm{tr}\left\{  \rho K_{l}^{\dagger}%
(\omega)K_{l}(\omega)\right\}  ,\label{10}\\
&  \rho_{out}(\omega|\rho)=\frac{\sum_{l}K_{l}(\omega)\rho K_{l}^{\dagger
}(\omega)}{\mu(\omega|\rho)}. \label{11}%
\end{align}
If, in representation (\ref{7}), a set $\{1,...,L\}$ contains only one
element, that is, Kraus operators are labeled only by outcomes $\omega
\in\Omega,$ then a state instrument $\mathcal{M}$ is called pure, since in
this case, mapping $\mathcal{M(\omega)[\cdot]}$ "transforms" a pure initial
state $|\psi\rangle\langle\psi|$ to the pure conditional posterior state
\begin{equation}
\rho_{out}(\omega|\text{ }|\psi\rangle\langle\psi|)=\frac{K(\omega
)|\psi\rangle\langle\psi|K^{\dagger}(\omega)}{\langle\psi|K^{\dagger}%
(\omega)K(\omega)|\psi\rangle},\mathit{\ \ \ }\omega\in\Omega. \label{12}%
\end{equation}

As proved by Ozawa \cite{m.ozawa}, for every observable instrument
$\mathcal{N}$, describing a generalized quantum measurement, there exists a
statistical realization
\begin{equation}
\Xi:=\{\widetilde{\mathcal{H}},\sigma,\mathrm{P},U\}, \label{13}%
\end{equation}
possibly, not unique, consisting of 4 elements: a complex Hilbert space
$\widetilde{\mathcal{H}},$ a density operator $\sigma$ on
$\widetilde{\mathcal{H}},$ a projection-valued measure $\mathrm{P}$ on
$\Omega$ with values $\mathrm{P}(\omega),\omega\in\Omega,$ that are
projections on $\widetilde{\mathcal{H}}$ and a unitary operator $U$ on
$\mathcal{H}\otimes\widetilde{\mathcal{H}}$, such that, for all $Y\in
\mathcal{L}\mathfrak{(}\widetilde{\mathcal{H}}),$%
\begin{equation}
\mathcal{N(}\omega)[Y]=\mathrm{tr}_{\widetilde{\mathcal{H}}}\left\{  \left(
\mathbb{I}_{\mathcal{H}}\otimes\sigma\right)  U^{\dagger}(Y\otimes
\mathrm{P}(\omega))U\right\}  ,\mathit{\ \ }\omega\in\Omega, \label{14}%
\end{equation}
where notation $\mathrm{tr}_{\widetilde{\mathcal{H}}}\{\cdot\}$ means the
partial trace over a Hilbert space $\widetilde{\mathcal{H}}.$ From relations
(\ref{2.1}) and (\ref{14}) it follows that, in terms of the elements of a
statistical realization (\ref{13}), the values $\mathcal{M(}\omega)[\cdot]$ of
a state instrument are given by%
\begin{equation}
\mathcal{M(}\omega)[\rho]=\mathrm{tr}_{\widetilde{\mathcal{H}}}\left\{
\left(  \mathbb{I}_{\mathcal{H}}\otimes\mathrm{P}(\omega)\right)
U(\rho\otimes\sigma)U^{\dagger}\left(  \mathbb{I}_{\mathcal{H}}\otimes
\mathrm{P}(\omega)\right)  \right\}  ,\mathit{\ \ }\omega\in\Omega, \label{15}%
\end{equation}
for all density operators $\rho$ on $\mathcal{H}.$

\begin{remark}
The existence for every generalized quantum measurement of a statistical
realization $\Xi=\{\widetilde{\mathcal{H}},\sigma,P,U\}$ means that each
generalized quantum measurement on a state $\rho$ on $\mathcal{H}$ can be
realized via the indirect measurement, specified by the elements of this
statistical realization $\Xi.$ Namely, via a direct measurement $P(\cdot)$ on
some auxiliary quantum system, being initially in a state $\sigma$ on a
Hilbert space $\widetilde{\mathcal{H}}$, after its interaction with the
original system which results in the composite system state $U(\rho
\otimes\sigma)U^{\dagger}$ on $\mathcal{H}\otimes\widetilde{\mathcal{H}}$.
\end{remark}

Representations (\ref{7}) and (\ref{15}) imply that if, for a state instrument
$\mathcal{M}$, there exists a statistical realization $\Xi$ where a state
$\sigma=|b\rangle\langle b|$ is pure while the values of a projection-valued
measure $\mathrm{P}$ have the form $\mathrm{P}(\omega)=|\xi_{\omega}%
\rangle\langle\xi_{\omega}|,$ where $\{|\xi_{\omega}\rangle,$ $\omega\in
\Omega\}$ is an orthonormal basis of $\widetilde{\mathcal{H}},$ then, for this
state instrument $\mathcal{M},$ there is representation (\ref{7}) where the
Kraus operators are labeled only by an outcome $\omega\in\Omega$ and are
defined via the relations (see Lemma 1 in \cite{e.r.loubenets}):
\begin{align}
\ \ \ \ \ \ \ \ \ \ \ \ \ \ \  &  U(|\psi\rangle\otimes|b\rangle)
=\sum_{\omega}K(\omega)|\psi\rangle\otimes|\xi_{\omega}\rangle,\text{
\ \ \ for each }|\psi\rangle\in\mathcal{H},\label{16}\\
&  \left(  \langle f|\otimes\langle\xi_{\omega}|\right)  U(|g\rangle
\otimes|b\rangle) =\langle f|K(\omega)g\rangle,\text{ \ \ \ \ \ for all
\ }|f\rangle,|g\rangle\in\mathcal{H}.\nonumber
\end{align}
In the physical notation,
\begin{equation}
K(\omega)=\langle\xi_{\omega}|U|b\rangle. \label{17}%
\end{equation}

\section{$N$-sequential conclusive discrimination between $r\geq2$ quantum
states}

For an $N$-sequential conclusive quantum state discrimination, introduce a
general scenario where a sender, say Alice, prepares a quantum system,
described in terms of a complex Hilbert space $\mathcal{H},$ \emph{possibly
infinite-dimensional,} in one of states $\rho_{1},...,\rho_{r},$ $r\geq2,$
pure or mixed, with a priori probabilities $q_{1},...,q_{r}$ and sends her
quantum system in the initial state
\begin{equation}
\rho_{in}=\sum_{j=1,...,r}q_{j}\rho_{j},\mathit{\ \ \ }\sum_{j}q_{j}%
=1,\mathit{\ \ \ }q_{j}>0, \label{18}%
\end{equation}
through a quantum channel to a chain of $N$ receivers. On receiving a quantum
system of Alice, the first receiver is allowed to perform on this system a
conclusive measurement for the discrimination between states $\{\rho
_{1},...,\rho_{r}\}$ while a quantum system of Alice in the posterior state,
conditioned by an outcome $j_{1}\in\{1,...,r\}$ observed by the first
receiver,\textbf{ }is further transmitted via a quantum channel to a
sequential receiver. This procedure is repeated until a conclusive measurement
by an $N$-th receiver.

We assume that either beforehand or during a sequential discrimination a
communication between receivers via a separate classical channel is prohibited
\cite{j.a.bergou}, the same concerns also a communication between receivers
via an extra quantum channel (by encoding the information on their outcomes
into orthogonal quantum states) -- since this allows the trivial strategy
where the first receiver performs an optimal measurement and communicates the
information on his outcome to sequential receivers.

Under the above constraints, the $N$-sequential conclusive state
discrimination scenario described below can provide a practical multipartite
quantum communication protocol.

Denote by $\mathcal{M}_{n}^{(r)}\mathcal{(}\cdot)\ (n=1,...,N)$ a state
instrument describing a conclusive quantum measurement of an $n$-th sequential
receiver with outcomes $j_{n}$ in set $\{1,...,r\}$. The corresponding
observable instrument $\mathcal{N}_{n}^{(r)}\mathcal{(}\cdot)$ and the POV
measure $\mathrm{M}_{n}^{(r)}$ with values in set $\{1,...,r\}$ are specified
by Eqs. (\ref{1.1})--(\ref{2.1}) in Section 2.

Let quantum channels between sequential receivers be ideal. The sequence
$A|\rightarrow1\rightarrow...\rightarrow k$ of quantum measurements performed
by $k\in\{1,...,N\}$ receivers, each with an outcome $j_{k}$ $\in\{1,...,r\}$,
constitutes a consecutive measurement with an outcome $\omega=\left(
j_{1},...,j_{k}\right)  \in\{1,...,r\}^{k}$ on the Alice quantum system in an
initial state (\ref{18}) and is described by the state instrument
$\mathcal{M}_{A|\rightarrow1\rightarrow...\rightarrow k}^{(r)}(\cdot)$ with
values \cite{a.s.holevo2}
\begin{align}
\mathcal{M}_{A|\rightarrow1\rightarrow...\rightarrow k}^{(r)}\left(
j_{1},...,j_{k}\right)  [\cdot]  &  :=\mathcal{M}_{k}^{(r)}\mathcal{(}%
j_{k})\left[  \mathcal{M}_{k-1}^{(r)}\mathcal{(}j_{k-1})\left[  \cdots
\mathcal{M}_{1}^{(r)}\mathcal{(}j_{1})\left[  \cdot\right]  ...\right]
\right]  ,\label{20}\\
&  \left(  j_{1},...,j_{k}\right)  \in\{1,...,r\}^{k}.\nonumber
\end{align}
By Eqs. (\ref{1.1})--(\ref{2.1}) the corresponding observable instrument
$\mathcal{N}_{A|\rightarrow1\rightarrow...\rightarrow k}^{(r)}$ and the POV
measure $\mathrm{M}_{A|\rightarrow1\rightarrow...\rightarrow k}$ of
${A|\rightarrow1\rightarrow...\rightarrow k}$ consecutive measurement
described by the state instrument (\ref{20}) have the forms%
\begin{align}
\mathcal{N}_{A|\rightarrow1\rightarrow...\rightarrow k}^{(r)}(j_{1}%
,...,j_{k})[\cdot]  &  :=\mathcal{N}_{1}^{(r)}\mathcal{(}j_{1})\left[
\mathcal{N}_{2}^{(r)}\mathcal{(}j_{2})\left[  \cdots\mathcal{N}_{N}%
^{(r)}\mathcal{(}j_{k})[\mathbb{\cdot}]...\right]  \right]  ,\label{21}\\
\mathrm{M}_{A|\rightarrow1\rightarrow...\rightarrow k}^{(r)}(j_{1},...,j_{k})
&  =\mathcal{N}_{A|\rightarrow1\rightarrow...\rightarrow k}^{(r)}%
(j_{1},...,j_{k})[\mathbb{I}_{\mathcal{H}}]. \label{22}%
\end{align}
From relations (\ref{20}) and (\ref{22}) it follows that, for an input state
(\ref{18}) before this consecutive measurement $A|\rightarrow1\rightarrow
...\rightarrow k$, the probability to receive under this measurement an
outcome $\omega=\left(  j_{1},...,j_{k}\right)  \in\{1,...,r\}^{k}$ is equal
to%
\begin{align}
\mu_{\mathcal{M}_{A|\rightarrow1\rightarrow...\rightarrow k}^{(r)}}%
(j_{1},...,j_{k}  &  \mid\rho_{in})=\sum_{j}q_{j}\mathrm{tr}\left\{
\mathcal{M}_{A|\rightarrow1\rightarrow...\rightarrow k}^{(r)}(j_{1}%
,...,j_{k})[\rho_{j}]\right\} \label{23}\\
&  =\sum_{j}q_{j}\mu_{\mathcal{M}_{A|\rightarrow1\rightarrow...\rightarrow
k}^{(r)}}(j_{1},...,j_{k}\mid\rho_{j}),\nonumber
\end{align}
where
\begin{align}
\mu_{\mathcal{M}_{A|\rightarrow1\rightarrow...\rightarrow k}^{(r)}}%
(j_{1},...,j_{k}  &  \mid\rho_{j})=\mathrm{tr}\left\{  \mathcal{M}%
_{A|\rightarrow1\rightarrow...\rightarrow k}^{(r)}(j_{1},...,j_{k})[\rho
_{j}]\right\} \label{23.1}\\
&  =\mathrm{tr}\{\rho_{j}\mathrm{M}_{A|\rightarrow1\rightarrow...\rightarrow
k}^{(r)}(j_{1},...,j_{k})\}.\nonumber
\end{align}

Therefore, within an $N$-sequential conclusive quantum state discrimination,
the probability for $N$ receivers to take the proper decisions on
discriminating between states $\rho_{1},...,\rho_{r},$ given with a priori
probabilities $q_{1},...,q_{r},$ \emph{i.e. the success probability }under a
consecutive measurement\emph{ }$A|\rightarrow1\rightarrow...\rightarrow N$,
has the form
\begin{align}
\mathbb{P}_{\mathcal{M}_{A|\rightarrow1\rightarrow...\rightarrow N}^{(r)}%
}^{success}(\rho_{1},...,\rho_{r}  &  \mid q_{1},...,q_{r})=\sum
_{j=1,...,r}q_{j}\mathrm{tr}\{\mathcal{M}_{A|\rightarrow1\rightarrow
...\rightarrow N}^{(r)}(\text{ }\overbrace{j,...,j}^{N})[\rho_{j}%
]\}\label{25}\\
&  =\sum_{j=1,...,r}q_{j}\mathrm{tr}\{\rho_{i}\mathrm{M}_{A|\rightarrow
1\rightarrow...\rightarrow N}^{(r)}(\text{ }\overbrace{j,...,j}^{N}%
)\},\nonumber
\end{align}
where $\mathrm{M}_{A|\rightarrow1\rightarrow...\rightarrow N}^{(r)}$ is the
POV measure (\ref{22}).

Besides representation (\ref{25}), let us also introduce two other equivalent
representations for the success probability $\mathrm{P}_{\mathcal{M}%
_{A|\rightarrow1\rightarrow...\rightarrow N}^{(r)}}^{success}$. Denote by
\begin{equation}
\tau_{out}^{(k)}(\overbrace{j,...,j}^{k}|\rho_{j}),\mathit{\ \ \ }k\geq1,
\label{26}%
\end{equation}
the posterior state on $\mathcal{H}$ conditioned by an outcome
$(\overbrace{j,...,j}^{k})\in\{1,...,r\}^{k}$ observed under the consecutive
measurement $\mathcal{M}_{A|\rightarrow1\rightarrow...\rightarrow k}$ on a
state $\rho_{j}$. In view of relations (\ref{5}), (\ref{20}) and (\ref{23}),
we have%
\begin{equation}
\tau_{out}^{(k)}(\overbrace{j,...,j}^{k}|\rho_{j})=\frac{\mathcal{M}%
_{A|\rightarrow1\rightarrow...\rightarrow(k-1)}^{(r)}(\overbrace{j,...,j)}%
^{k}[\rho_{j}]}{\mu_{\mathcal{M}_{A|\rightarrow1\rightarrow...\rightarrow
(k-1)}^{(r)}}(\underbrace{j,...,j}_{k}|\rho_{j})} \label{27}%
\end{equation}
and%
\begin{align}
&  \mathcal{M}_{A|\rightarrow1\rightarrow...\rightarrow k}^{(r)}%
(\overbrace{j,...,j)}^{k}[\rho_{j}]=\mathcal{M}_{k}^{(r)}\mathcal{(}j)\left[
\mathcal{M}_{A|\rightarrow1\rightarrow...\rightarrow(k-1)}^{(r)}%
(\overbrace{j,...,j)}^{k-1}[\rho_{j}]\right] \label{28}\\
&  =\mu_{\mathcal{M}_{A|\rightarrow1\rightarrow...\rightarrow(k-1)}^{(r)}%
}(\overbrace{j,...,j}^{k-1}|\rho_{i})\times\mathcal{M}_{k}^{(r)}%
\mathcal{(}j)[\tau_{out}^{(k-1)}(\overbrace{j,...,j}^{(k-1)}|\rho
_{j})].\nonumber
\end{align}
Since a length of a tuple $(j,...,j)\in\{1,...,r\}^{k}$ is equal to a number
$k$ standing in indices at the probability $\mu$ and at the posterior state
$\tau_{out}^{(k)}$, for short, we below omit the upper decoration at
$\overbrace{j,...,j}^{k}.$

By Eq. (\ref{28}), for all $k\geq2,$ we come to the following representation
of the success probability $\mathrm{P}_{\mathcal{M}_{A|\rightarrow
1\rightarrow...\rightarrow k}^{(r)}}^{success}$ under a consecutive
measurement\emph{ }$A|\rightarrow1\rightarrow...\rightarrow k:$%
\begin{align}
\text{ }\mathbb{P}_{\mathcal{M}_{A|\rightarrow1\rightarrow...\rightarrow
k}^{(r)}}^{success}(\rho_{1},...,\rho_{r}  &  \mid q_{1},...,q_{r}%
)\label{29}\\
&  =\sum_{j=1,...,r}q_{j}\mu_{\mathcal{M}_{A|\rightarrow1\rightarrow
...\rightarrow(k-1)}^{(r)}}(j,...,j|\rho_{j})\times\mathrm{tr}\left\{
\mathcal{M}_{k}^{(r)}\mathcal{(}j)\left[  \tau_{out}^{(k-1)}(j,...,j|\rho
_{j})\right]  \right\}  .\nonumber
\end{align}
Relations (\ref{23.1}) and (\ref{29}) imply\ the following general statement.

\begin{proposition}
Under an $N$-sequential conclusive discrimination between $r\geq2$ arbitrary
states $\rho_{1},...,\rho_{r}$, described by a state instrument $\mathcal{M}%
_{A|\rightarrow1\rightarrow...\rightarrow N}^{(r)},$ $N\geq2,$ the success
probability\textbf{ }$\mathbb{P}_{\mathcal{M}_{A|\rightarrow1\rightarrow
...\rightarrow N}^{(r)}}^{success}$ admits representation (\ref{25}) and also
the following representations equivalent to (\ref{25}):%
\begin{align}
&  \mathbb{P}_{\mathcal{M}_{A|\rightarrow1\rightarrow...\rightarrow N}^{(r)}%
}^{success}(\rho_{1},...,\rho_{r}\mid q_{1},...,q_{r})\label{30}\\
&  =\sum_{j=1,...,r}q_{j}\mathrm{tr}\{\rho_{j}\mathrm{M}_{1}^{(r)}(j)\}\text{
}\mathrm{tr}\{\tau_{out}^{(1)}(j|\rho_{j})\mathrm{M}_{2}^{(r)}(j)\}\times
\cdots\times\mathrm{tr}\{\tau_{out}^{(N-1)}(j,...,j\text{ }|\rho
_{j})\mathrm{M}_{N}^{(r)}(j)\},\nonumber\\
& \nonumber\\
&  \mathbb{P}_{\mathcal{M}_{A|\rightarrow1\rightarrow...\rightarrow N}^{(r)}%
}^{success}(\rho_{1},...,\rho_{r}\mid q_{1},...,q_{r})=\mathbb{P}%
_{\mathcal{M}_{1}^{(r)}}^{success}(\rho_{1},...,\rho_{r}\mid q_{1}%
,...,q_{r})\label{31}\\
&  \times%
{\displaystyle\prod\limits_{n=2,...,N}}
\mathbb{P}_{\mathcal{M}_{n}^{(r)}}^{success}\left(  \tau_{out}^{(n-1)}%
(1,...,1|\rho_{1}),...,\tau_{out}^{(n-1)}(r,...,r|\rho_{r})\mid Q_{n-1}%
^{(1)},...,Q_{n-1}^{(r)}\right)  .\nonumber
\end{align}
Here: (i) $\mathrm{M}_{n}^{(r)}(j):=\mathcal{N}_{n}(j)[\mathbb{I}%
_{\mathcal{H}}]$ is the POV measure describing a measurement of each
sequential receiver; (ii) $\tau_{out}^{(n-1)}(j,...,j|\rho_{j}),$ $n\geq2,$ is
the posterior state after $(n-1)$-th measurement, each originated from a state
$\rho_{j}$ and conditioned by the outcome \textquotedblleft$j$%
\textquotedblright\ under the conclusive measurements of all previous
receivers; (iii)
\begin{equation}
Q_{n-1}^{(j)}:=\frac{q_{n-1}\mu_{\mathcal{M}_{A|\rightarrow1\rightarrow
...\rightarrow(n-1)}^{(r)}}\left(  j,...,j|\rho_{j}\right)  }{\mathbb{P}%
_{\mathcal{M}_{A|\rightarrow1\rightarrow...\rightarrow(n-1)}^{(r)}}%
^{success}(\rho_{1},...,\rho_{r}\mid q_{1},...,q_{r})},\mathit{\ \ \ }%
\sum_{j=1,...,r}Q_{n-1}^{(j)}=1,\mathit{\ \ }n\geq2, \label{32}%
\end{equation}
is a priori probability of the posterior state $\tau_{out}^{(n-1)}%
(j,...,j|\rho_{j})$ before a measurement of each $n\geq2$ sequential receiver.
\end{proposition}

From representations (\ref{30}), (\ref{31}) and the general upper bound
derived in \cite{loubenets} for the success probability in case of only one
receiver%
\begin{equation}
\mathbb{P}_{\mathcal{M}_{A|\rightarrow1}^{(r)}}^{success}(\rho_{1}%
,...,\rho_{r}|q_{1},....,q_{r})\leq\frac{1}{r}\left(  1+%
{\displaystyle\sum\limits_{1\leq i<j\leq r}}
\left\Vert \left(  q_{i}\rho_{i}-q_{j}\rho_{j}\right)  \right\Vert
_{1}\right)  , \label{32.1}%
\end{equation}
it follows that, for any protocol of the $N$-sequential conclusive
discrimination between $r\geq2$ quantum states $\rho_{1},...,\rho_{r},$ the
success probability admits the upper bound%
\begin{equation}
\mathbb{P}_{\mathcal{M}_{A|\rightarrow1\rightarrow...\rightarrow N}^{(r)}%
}^{success}(\rho_{1},...,\rho_{r}\mid q_{1},...,q_{r})\leq\frac{1}{r}\left(
1+%
{\displaystyle\sum\limits_{1\leq i<j\leq r}}
\left\Vert \left(  q_{i}\rho_{i}-q_{j}\rho_{j}\right)  \right\Vert
_{1}\right)  . \label{32.2}%
\end{equation}
Here, notation $\left\Vert \cdot\right\Vert _{1}$ means the trace norm of the
Hermitian operator $\left(  q_{i}\rho_{i}-q_{j}\rho_{j}\right)  .$ Recall
that, for any bounded Hermitian operator,
\begin{align}
X  &  =X^{(+)}-X^{(-)},\text{ \ \ }X^{(+)},X^{(-)}\geq0,\label{32.3}\\
\left\Vert X\right\Vert _{1}  &  =\left\Vert X^{(+)}\right\Vert _{1}%
+\left\Vert X^{(-)}\right\Vert _{1},\nonumber\\
\left\Vert X^{(\pm)}\right\Vert _{1}  &  =\mathrm{tr}\{X^{(\pm)}\}.\nonumber
\end{align}

We stress that, in a general scenario for an $N$-sequential conclusive quantum
state discrimination, which we have introduced above, no any constraints are
imposed on discriminated states $\rho_{1},...,\rho_{r}$ and a prior
probabilities $q_{1},...,q_{r}.$

\section{Optimal protocols}

Let $\mathfrak{M}_{r,N}^{(cond)}$ be a set of quantum state instruments
$\mathcal{M}_{A|\rightarrow1\rightarrow...\rightarrow N}^{(r)}$, describing an
$N$-sequential measurement with outcomes in $\{1,...,r\}^{N}$ under some extra
condition on receivers' quantum measurements which is determined by some
scenario strategy\footnote{This is, for example, the case in \cite{d.fields}
where receivers' quantum measurements are described by specific quantum
instruments.}. The optimal success probability within this strategy is defined
by the maximum%
\begin{equation}
\mathbb{P}_{A|\rightarrow1\rightarrow...\rightarrow N}^{opt.success}\left(
\rho_{1},...,\rho_{r}\mid q_{1},...,q_{r}\right)  \text{ }|_{\mathfrak{M}%
_{r,N}^{(cond)}}=\max_{\mathfrak{M}_{r,N}^{(cond)}}\text{ }\mathbb{P}%
_{\mathcal{M}_{A|\rightarrow1\rightarrow...\rightarrow N}^{(r)}}%
^{success}(\rho_{1},...,\rho_{r}\mid q_{1},...,q_{r}), \label{33}%
\end{equation}
where success probabilities $\mathrm{P}_{\mathcal{M}_{A|\rightarrow
1\rightarrow...\rightarrow N}^{(r)}}^{success}$ are given by either of the
equivalent representations (\ref{25}), (\ref{30}), (\ref{31}). If no any extra
condition on receivers' quantum measurements is put, then the optimal success
probability is given by
\begin{equation}
\mathbb{P}_{A|\rightarrow1\rightarrow...\rightarrow N}^{opt.success}\left(
\rho_{1},...,\rho_{r}\mid q_{1},...,q_{r}\right)  =\max_{\mathfrak{M}_{r,N}%
}\mathbb{P}_{\mathcal{M}_{A|\rightarrow1\rightarrow...\rightarrow N}^{(r)}%
}^{success}(\rho_{1},...,\rho_{r}\mid q_{1},...,q_{r}) \label{33.1}%
\end{equation}
where $\mathfrak{M}_{r,N}=\{\mathfrak{M}_{r,N}^{(cond)}\}$ is the convex set
of all possible state instruments $\mathcal{M}_{A|\rightarrow1\rightarrow
...\rightarrow N}^{(r)}$ describing a consecutive measurement of $N$
receivers. Clearly,
\begin{equation}
\mathbb{P}_{A|\rightarrow1\rightarrow...\rightarrow N}^{opt.success}%
|_{\mathfrak{M}_{r,N}^{(cond)}}\text{ }\leq\text{ }\mathbb{P}_{A|\rightarrow
1\rightarrow...\rightarrow N}^{opt.success}.
\end{equation}

Since in the product of the success probabilities standing in (\ref{31}) the
success probability of each receiver depends on measurement parameters of the
previous receivers, in a general case, the optimal success probability in
(\ref{33.1}) does not need to be equal to the product of the optimal success
probabilities of $N$ receivers.

\subsection{Optimal success probability}

From relations (\ref{31}), (\ref{32.1}) and (\ref{33.1}) it follows that, for
an arbitrary number $r\geq2$ of quantum states states $\rho_{1},...,\rho_{r}$,
pure and mixed and any a priori probabilities $q_{1},...,q_{r}$, the optimal
success probability%
\begin{align}
\mathbb{P}_{A|\rightarrow1\rightarrow...\rightarrow N}^{opt.success}\left(
\rho_{1},...,\rho_{r}\mid q_{1},...,q_{r}\right)   &  \leq\mathbb{P}%
_{A|\rightarrow1}^{opt.success}(\rho_{1},...,\rho_{r}\mid q_{1},...,q_{r}%
)\label{36}\\
&  \leq\frac{1}{r}\left(  1+%
{\displaystyle\sum\limits_{1\leq i<j\leq r}}
\left\Vert q_{i}\rho_{i}-q_{j}\rho_{j}\right\Vert _{1}\right)  .\nonumber
\end{align}
If $N$ receivers discriminate between two states $(r=2)$, then the upper bound
in the second line of (\ref{36}) reduces to the Helstrom upper bound
\cite{c.w.helstrom}:
\begin{align}
\mathbb{P}_{A\rightarrow1\rightarrow...\rightarrow N}^{opt.success}(\rho
_{1},\rho_{2}  &  \mid q_{1},q_{2})\leq\mathbb{P}_{A|\rightarrow
1}^{opt.success}(\rho_{1},\rho_{2}\mid q_{1},q_{2})\label{36.1}\\
&  =\frac{1}{2}\left(  1+\left\Vert q_{1}\rho_{1}-q_{2}\rho_{2}\right\Vert
_{1}\right)  .\nonumber
\end{align}
which is attained on\emph{ any state instrument} $\mathcal{M}_{1}^{(2)}$ (see
(\ref{37})) of the first receiver with\footnote{One and the same POV measure
may correspond to different state instruments, see in Section 2.} the POV
measure
\begin{equation}
\mathrm{M}_{opt}^{(2)}(j)=\mathrm{P}_{0}^{(2)}(j),\text{ \ \ }j=1,2,
\label{36.2}%
\end{equation}
which is projection-valued and defined by:%
\begin{align}
\mathrm{P}_{0}^{(2)}(1)  &  =\sum_{\lambda_{k}\text{ }>0,}\mathrm{E(}%
\lambda_{k}),\mathit{\ \ }\mathrm{P}_{0}^{(2)}(2)=\mathbb{I}_{\mathcal{H}%
}-\mathrm{P}_{0}^{(2)}(1),\label{38}\\
\left(  \mathrm{P}_{0}^{(2)}(j)\right)  ^{2}  &  =\mathrm{P}_{0}%
^{(2)}(j),\text{ }\mathit{\ \ }j=1,2,\text{ \ \ }\mathrm{P}_{0}^{(2)}%
(1)\mathrm{P}_{0}^{(2)}(2)=\mathrm{P}_{0}^{(2)}(2)\mathrm{P}_{0}%
^{(2)}(1)=0.\nonumber
\end{align}
Here, $\mathrm{E(}\lambda_{k})$ are the spectral projections of the Hermitian
operator $q_{1}\rho_{1}-q_{2}\rho_{2}=\sum_{\lambda_{k}}\lambda_{k}%
\mathrm{E(}\lambda_{k}),$ $\sum_{\lambda_{k}}\mathrm{E(}\lambda_{k})=$
$\mathbb{I}_{\mathcal{H}}$, and $\mathrm{P}_{0}^{(2)}(1)$ is the orthogonal
projection on the invariant subspace of operator $(q_{1}\rho_{1}-q_{2}\rho
_{2}),$ corresponding to its positive eigenvalues.

\begin{remark}
We stress that, in case of discrimination between $r>2$ arbitrary quantum
states $\rho_{1},...,\rho_{r}$, the precise expression for the optimal
POV\ measure is not known even for only one receiver, and, after the optimal
measurement of the first receiver, the conditional posterior states
corresponding to different outcomes do not need to be mutually orthogonal.
\end{remark}

However, the following statement introduces the special condition on $r\geq2$
quantum states where under, the $N$-sequential conclusive quantum state
discrimination scenario specified in Section 3, the upper bound in the first
line of (\ref{36}) is attained.

\begin{theorem}
Let $\rho_{1},...,\rho_{r},$ $r\geq2,$ be arbitrary quantum states, pure or
mixed, on a Hilbert space $\mathcal{H}$ of dimension $d\geq r$, given with a
priori probabilities $q_{1},...,q_{r},$ and $\mathrm{M}_{opt}^{(r)}$ be the
optimal POV measure of the first receiver measurement for the discrimination
between these states:%
\begin{align}
\mathbb{P}_{A|\rightarrow1}^{opt.success}(\rho_{1},...,\rho_{r}  &  \mid
q_{1},...,q_{r})=\max_{\mathcal{M}_{1}^{(2)}}\mathbb{P}_{\mathcal{M}_{1}%
^{(2)}}^{success}(\rho_{1},...,\rho_{r}\mid q_{1},...,q_{r})\label{37}\\
&  =\sum_{j=1,...,r}q_{j}\mathrm{tr}\left\{  \rho_{j}\mathrm{M}_{opt}%
^{(r)}(j)\right\}  .\nonumber
\end{align}
If the optimal measurement of the first receiver on states $\rho_{1}%
,...,\rho_{r}$ can be realized via a quantum state instrument\footnote{See
representation (\ref{7}).}
\begin{equation}
\mathcal{M}_{1}^{(r)}=\sum_{j=1,...,r}\mathrm{K}_{r}(j)[\cdot]\mathrm{K}%
_{r}^{\dagger}(j),\text{ \ \ \ \ \ \ \ \ }\mathrm{K}_{r}^{\dagger
}(j)\mathrm{K}_{r}(j)=\mathrm{M}_{opt}^{(r)}(j), \label{38_}%
\end{equation}
where the Kraus operators $\mathrm{K}_{r},$ $j=1,...,r,$ satisfy the relation%
\begin{equation}
\mathrm{K}_{r}^{\dag}(j)\mathrm{P}^{(r)}(j)\mathrm{K}_{r}(j)=\mathrm{M}%
_{opt}^{(r)}(j),\text{ \ \ \ \ }j=1,...,r, \label{44.1}%
\end{equation}
for some mutually orthogonal projections $\mathrm{P}^{(r)}(j)$ on
$\mathcal{H}:$
\begin{equation}
\sum_{j=1,...,r}\mathrm{P}^{(r)}(j)=\mathbb{I}_{\mathcal{H}},\text{
\ \ \ }\mathrm{P}^{(r)}(j_{n_{1}})\mathrm{P}^{(r)}(j_{n_{2}})=\delta
_{n_{1}n_{2}}\mathrm{P}^{(r)}(j_{n_{1}}), \label{44_1}%
\end{equation}
then, under a conclusive state discrimination with any number $N\geq2$ of the
sequential receivers, the optimal success probability (\ref{33.1}) equals to
the optimal success probability (\ref{37}) of the first receiver:%
\begin{equation}
\mathbb{P}_{A|\rightarrow1\rightarrow...\rightarrow N}^{opt.success}(\rho
_{1},...,\rho_{r}|q_{1},...,q_{r})=\mathbb{P}_{A|\rightarrow1}^{opt.success}%
(\rho_{1},...,\rho_{r}|q_{1},...,q_{r}),
\end{equation}
and is attained under the protocol, described by the quantum state instrument%
\begin{align}
&  \mathfrak{J}_{A|\rightarrow1\rightarrow...\rightarrow N}^{(r)}\left(
j_{1},...,j_{N}\right)  [\cdot]\label{44.2}\\
&  =\mathrm{P}^{(r)}(j_{N})\cdot\ldots\cdot\mathrm{P}^{(r)}(j_{2})\text{
}\mathrm{K}_{r}(j_{1})[\cdot]\text{ }\mathrm{K}_{r}^{\dag}(j_{1})\text{
}\mathrm{P}^{(r)}(j_{2})\cdot\ldots\cdot\mathrm{P}^{(r)}(j_{N}),\nonumber\\
j_{n}  &  =1,...,r.\nonumber
\end{align}

\end{theorem}

\begin{proof}
For the state instrument (\ref{44.2}), we have%
\begin{equation}
\mathfrak{J}_{A|\rightarrow1\rightarrow...\rightarrow N}^{(r)}\left(
j,...,j\right)  [\cdot]=\mathrm{P}^{(r)}(j)\text{ }\mathrm{K}_{r}%
(j)[\cdot]\mathrm{K}_{r}^{\dag}(j)\text{ }\mathrm{P}^{(r)}(j),\text{
\ \ }\mathit{\ }j=1,...,r.
\end{equation}
In view of (\ref{2.1}), (\ref{25}), (\ref{37}), (\ref{38_}) and (\ref{44.1}),
this implies that, under the discrimination protocol described by the state
instrument (\ref{44.2}), we have the following relations:%
\begin{align}
\mathbb{P}_{\mathfrak{J}_{A|\rightarrow1\rightarrow...\rightarrow N}^{(2)}%
}^{success}(\rho_{1},...,\rho_{r}|q_{1},...,q_{r})  &  =\sum_{j=1,...,r}%
q_{j}\mathrm{tr}\left\{  \mathfrak{J}_{A|\rightarrow1\rightarrow...\rightarrow
N}^{(r)}\left(  j,...,j\right)  [\rho_{j}]\right\} \\
&  =\sum_{j=1,...,r}q_{j}\mathrm{tr}\left\{  \rho_{j}\mathrm{K}_{r}^{\dagger
}(j)\text{ }\mathrm{P}^{(r)}(j)\text{\textrm{ }}\mathrm{K}_{r}(j)\right\}
\nonumber\\
&  =\sum_{j=1,...,r}q_{j}\mathrm{tr}\left\{  \rho_{j}\mathrm{M}_{opt}%
^{(\rho_{1},...,\rho_{r})}(j)\right\} \nonumber\\
&  =\mathbb{P}_{A|\rightarrow1}^{opt.success}(\rho_{1},...,\rho_{r}%
|q_{1},...,q_{r}).\nonumber
\end{align}
This proves the statement.
\end{proof}

\begin{remark}
Within the optimal protocol (\ref{44.2}), after the optimal measurement
(\ref{38_}) of the first receiver, the conditional posterior states
corresponding to different outcomes do not need to be mutually orthogonal,
however, due to condition (\ref{44.1}) and updated in view of (\ref{32}) a
priori probabilities for these posterior states, the success probability of
the second receiver is equal to $1.$
\end{remark}

For $r=2$, the relations (\ref{38_}) and (\ref{44.1}) are true for any two
quantum states and the quantum state instrument of the first receiver, given
by:%
\begin{equation}
\mathcal{M}_{1}^{(2)}[\cdot]=\sum_{j=1,2}\mathrm{P}_{0}^{(2)}(j)[\cdot
]\mathrm{P}_{0}^{(2)}(j) \label{49_1}%
\end{equation}
where the Kraus operators $\mathrm{P}_{0}^{(2)}(j),$ $j=1,2,$ are given in
(\ref{38}) and, in (\ref{44.1}), the set of the orthogonal projections is
$\left\{  \mathrm{P}_{0}^{(2)}(1),\mathrm{P}_{0}^{(2)}(2)\right\}  .$ If two
quantum states $\rho_{1},\rho_{2}$ are such that, in (\ref{38}), each
projection $\mathrm{P}_{0}^{(2)}(j)\neq0$, $j=1,2,$ then, for these two
quantum states, relations (\ref{38_}) and (\ref{44.1}) are also true under the
quantum state instrument of the form
\begin{equation}
\widetilde{\mathcal{M}}_{1}^{(2)}[\cdot]=\sum_{j=1,2}\widetilde{\mathrm{K}%
}_{2}(j)[\cdot]\widetilde{\mathrm{K}}_{2}^{\dagger}(j), \label{49-1}%
\end{equation}
where the Kraus operators $\widetilde{\mathrm{K}}_{2}(j)$ are defined via the
relations
\begin{align}
\widetilde{\mathrm{K}}_{2}(j)  &  =\sum_{i=1,...,k(j)}|\phi_{i}(j)\rangle
\langle\mathrm{v}_{i}(j)|,\text{ \ \ \ }\sum_{i=1,...,k(j)}|\mathrm{v}%
_{i}(j)\rangle\langle\mathrm{v}_{i}(j)|\text{ }=\mathrm{P}_{0}^{(2)}(j),\text{
\ }\ j=1,2,\label{49_2}\\
\left\langle \text{ }\mathrm{v}_{i_{1}}(j_{1})\text{ }|\text{ }\mathrm{v}%
_{i_{2}}(j_{2})\right\rangle  &  =\left\langle \text{ }\phi_{i_{1}}%
(j_{1})\text{ }|\text{ }\phi_{i_{2}}(j_{2})\right\rangle =\delta_{i_{1}i_{2}%
}\delta_{j_{1}j_{2}},\nonumber
\end{align}
while the orthogonal projections in (\ref{44.1}) are given by
\begin{equation}
\widetilde{\mathrm{P}}^{(2)}(j)=\sum_{i=1,...,k(j)}|\phi_{i}(j)\rangle
\langle\phi_{i}(j)|,\text{ \ }\ j=1,2. \label{49_3}%
\end{equation}
Here, (i) $\{|\mathrm{v}_{i}(j)\rangle,i=1,...,k(j)\},$ $j=1.2,$ are any two
orthonormal bases spanning the invariant subspaces of of the Hermitian
operator $(q_{1}\rho_{1}-q_{2}\rho_{2}),$ corresponding to its positive and
nonpositive eigenvalues, respectively, and having dimensions $k(j),$ $j=1,2,$
$\sum\limits_{j=1,2}k(j)=d$, and (ii) $\left\{  \text{ }|\phi_{i}%
(j)\rangle,\text{ }i=1,...,k(j),\ \ j=1,2\right\}  $ is any orthonormal basis
of a Hilbert space $\mathcal{H}$. If\ the bases
\begin{equation}
\{|\phi_{i}(j)\rangle,\text{ }i=1,...,k(j),\text{ }j=1,2\}=\{|\mathrm{v}%
_{i}(j)\rangle,\text{ }i=1,...,k(j),\text{ }j=1,2\},
\end{equation}
then $\widetilde{\mathrm{K}}_{2}(j)=\widetilde{\mathrm{P}}^{(2)}%
(j)=\mathrm{P}_{0}^{(2)}(j)$ and the quantum state instrument (\ref{49-1})
reduces to the state instrument (\ref{49_1}).

Therefore, for an $N$-sequential conclusive discrimination between two quantum
states $\rho_{1},\rho_{2}$ the statement of Theorem 1 implies.

\begin{theorem}
Let $\rho_{1},\rho_{2},$ given with a priori probabilities $q_{1},q_{2},$ be
two arbitrary quantum states, pure or mixed, on a Hilbert space $\mathcal{H}$
of an arbitrary dimension $d\geq2$. For an $N$-sequential conclusive
discrimination between two states $\rho_{1},\rho_{2}$, the optimal success
probability (\ref{33.1}) is equal to the Helstrom bound
\begin{equation}
\mathbb{P}_{A|\rightarrow1\rightarrow...\rightarrow N}^{opt.success}\left(
\rho_{1},\rho_{2}|q_{1},q_{2}\right)  =\frac{1}{2}\left(  1+\left\Vert
q_{1}\rho_{1}-q_{2}\rho_{2}\right\Vert _{1}\right)  , \label{40}%
\end{equation}
for every number $N\geq1$ of sequential receivers and is attained under the
optimal protocol described by the quantum state instrument
\begin{align}
\mathcal{M}_{A|\rightarrow1\rightarrow...\rightarrow N}^{(2)}\left(
j_{1},...,j_{N}\right)  [\cdot]  &  :=\mathrm{P}_{0}^{(2)}(j_{N})\cdot
\ldots\cdot\mathrm{P}_{0}^{(2)}(j_{1})[\cdot]\mathrm{P}_{0}^{(2)}(j_{1}%
)\cdot\ldots\cdot\mathrm{P}_{0}^{(2)}(j_{N}),\label{39}\\
j_{n}  &  =1,2.\nonumber
\end{align}
If quantum states $\rho_{1},\rho_{2}$ are such that $\mathrm{P}_{0}%
^{(2)}(j)\neq0$, $j=1,2,$ then, for these states, the optimal success
probability (\ref{40}) is also attained under either of the optimal protocols
described by the quantum state instrument of the form:%
\begin{align}
&  \widetilde{\mathcal{M}}_{A|\rightarrow1\rightarrow...\rightarrow N}%
^{(2)}\left(  j_{1},...,j_{N}\right)  [\cdot]\label{39_1}\\
&  =\widetilde{\mathrm{P}}^{(2)}(j_{N})\cdot\ldots\cdot\widetilde{\mathrm{P}%
}^{(2)}(j_{2})\text{ }\widetilde{\mathrm{K}}_{2}(j_{1})[\cdot]\text{
}\widetilde{\mathrm{K}}_{2}^{\dag}(j_{1})\text{ }\widetilde{\mathrm{P}}%
^{(2)}(j_{2})\cdot\ldots\cdot\widetilde{\mathrm{P}}^{(2)}(j_{N}),\nonumber\\
\widetilde{\mathrm{P}}^{(2)}(j_{n})\mathbf{\ }  &  \neq\mathrm{P}_{0}%
^{(2)}(j_{n}),\text{ \ \ \ }j_{n}=1,2.\nonumber
\end{align}
Here, (i) $\mathrm{P}_{0}^{(2)}(j),$\textbf{ }$j=1,2,$ are orthogonal
projections in (\ref{38}); (ii) the Kraus operators $\widetilde{\mathrm{K}%
}_{2}(j),$ $j=1,2$ are defined by (\ref{49_2}); (iii) the orthogonal
projections $\widetilde{\mathrm{P}}^{(2)}(j)$ are given in (\ref{49_3}).
\end{theorem}

Under each of the optimal protocols described by the state instruments
(\ref{39}) and (\ref{39_1}), the measurement of the first receiver for the
discrimination between any two quantum states $\rho_{1},\rho_{2}$ is the
optimal one described by projections in (\ref{38}) while if, in (\ref{38}),
the orthogonal projections $\mathrm{P}_{0}^{(2)}(j)\neq0$, $j=1,2,$ then,
after the measurement of the first receiver, the conditional posterior states,
corresponding to different outcomes $j=1,2$, are mutually orthogonal (but not,
in general, pure), otherwise, are given by $\rho_{1}$ (or $\rho_{2}$).

\subsection{Implementation via indirect measurements}

Let, under an $N$-sequential conclusive discrimination between two arbitrary
quantum states, pure or mixed, all receivers perform indirect measurements.

\begin{proposition}
The optimal $N$-sequential conclusive state discrimination protocol specified
in Theorem 2 is implemented if each $n$-th receiver performs the indirect
measurement described by the statistical realization\footnote{See in Section
2.}
\begin{equation}
\Xi_{n}=\left\{  \mathbb{C}^{2},|b_{n}\rangle\langle b_{n}%
|,\widetilde{\mathrm{P}}_{n},U_{n}\right\}  , \label{45}%
\end{equation}
where $|b_{n}\rangle\langle b_{n}|$ is a pure state on $\mathbb{C}^{2}$, the
projection-valued measure $\widetilde{\mathrm{P}}_{n}$ on $\mathbb{C}^{2}$ has
the elements $\widetilde{\mathrm{P}}_{n}(1)=|b_{n}\rangle\langle b_{n}|$ and
$\widetilde{\mathrm{P}}_{n}(2)=|b_{n}^{\bot}\rangle\langle b_{n}^{\bot
}|=\mathbb{I}_{\bar{\mathbb{C}^{2}}}-|b_{n}\rangle\langle b_{n}|$, and $U_{n}$
is the unitary operator on $\mathcal{H}\otimes\mathbb{C}^{2}$ having the
CNOT-like form:
\begin{equation}
U_{n}=\mathrm{P}_{0}(1)\otimes\mathbb{I}_{\mathbb{C}^{2}}+\mathrm{P}%
_{0}(2)\otimes\left(  \text{ }|b_{n}^{\bot}\rangle\langle b_{n}|\text{
}+\text{ }|b_{n}\rangle\langle b_{n}^{\bot}|\right)  . \label{46}%
\end{equation}
Here, $\mathrm{P}_{0}(j),$ $j=1,2$ are projections (\ref{38}) on a Hilbert
space $\mathcal{H}$.
\end{proposition}

\begin{proof}
Let the indirect measurement of each $n$-th receiver be described by the
statistical realization (\ref{45}). Then by (\ref{15}), the state instrument
$\mathcal{M}_{n}^{(2)}(\cdot)$ describing the $n$-th receiver's indirect
measurement (\ref{45}) is given by
\begin{align}
\mathcal{M}_{n}^{(2)}(j)[\sigma]  &  =\mathrm{tr}_{\bar{\mathcal{H}}}\left\{
\mathit{\ }\mathbb{I}_{\mathcal{H}}\otimes\widetilde{\mathrm{P}}_{n}(j)\left(
\text{ }U(\sigma\otimes|b_{n}\rangle\langle b_{n}|)U^{\dagger}\right)
\mathbb{I}_{\mathcal{H}}\otimes\widetilde{\mathrm{P}}_{n}(j)\right\}
\label{47}\\
&  =\mathrm{P}_{0}(j)\text{ }[\sigma]\text{ }\mathrm{P}_{0}(j),\mathit{\ \ \ }%
j=1,2,\mathit{\ \ \ }\forall\sigma,\nonumber
\end{align}
and constitutes the $N$-sequential conclusive discrimination protocol
described in (\ref{39}). According to Theorem 2, this proves the statement of
Proposition 3.
\end{proof}

\section{$N$-sequential conclusive discrimination under a noise}

In Section 3, we have introduced the general framework for the description of
an $N$-sequential conclusive state discrimination in case of an ideal
multipartite quantum communication between participants. However, in a
realistic situation, a multipartite quantum communication between participants
is noisy and, in this Section, we proceed to specify a general framework for
the description of an $N$-sequential conclusive state discrimination via noisy
quantum communication channels.

Let $\Lambda_{1}$ be a noisy quantum channel between a sender and a first
receiver and $\Lambda_{n}$ -- between an $(n-1)$-th receiver and an $n$-th
receiver. In this case, similarly to (\ref{20}), a quantum instrument
describing a consecutive measurement with an outcome $(j_{1},...,j_{k}%
)\in\{1,...,r\}^{k}$ performed by $N$ participants communicated via noisy
quantum channels takes the form
\begin{align}
\overset{\leadsto}{\mathcal{M}}_{A|\leadsto1\leadsto...\leadsto k}%
^{(r)}\left(  j_{1},...,j_{k}\right)  [\cdot]  &  :=\mathcal{M}_{k}%
^{(r)}\mathcal{(}j)\left[  \Lambda_{k}\left[  \mathcal{M}_{k-1}^{(r)}%
\mathcal{(}j_{k})\left[  \cdots\mathcal{M}_{1}^{(r)}\mathcal{(}j_{1})\left[
\Lambda_{1}[\cdot]\right]  \cdots\right]  \right]  \right]  ,\label{47.1}\\
&  \left(  j_{1},...,j_{k}\right)  \in\{1,...,r\}^{k},\nonumber
\end{align}
Here, symbol $\leadsto$ means a quantum state discrimination under a noisy
communication between participants.\textbf{ }

From (\ref{47.1}) it follows that, under a noisy multipartite communication,
representation (\ref{25}) for the success probability is to be replaced by%

\begin{align}
&  \mathbb{P}_{\overset{\leadsto}{\mathcal{M}}_{A|\rightsquigarrow
1\leadsto..\leadsto N}^{(r)}}^{success}\left(  \text{ }\rho_{1},...,\rho
_{r}|q_{1},...,q_{r}\right)  |_{\Lambda_{1},...,\Lambda_{_{N}}}\label{48}\\
&  =\sum_{j=1,...,r}q_{j}\mathrm{tr}\left\{  \mathcal{M}_{N}^{(r)}%
\mathcal{(}j)\left[  \Lambda_{N}\left[  \mathcal{M}_{N-1}^{(r)}\mathcal{(}%
j)\left[  \cdots\mathcal{M}_{1}^{(r)}\mathcal{(}j)[\Lambda_{1}[\rho_{j}%
]\cdots\right]  \right]  \right]  \right\} \nonumber\\
&  =\sum_{j=1,...,r}q_{j}\mu_{\overset{\leadsto}{\mathcal{M}}_{A|\leadsto
1\leadsto...\leadsto N}^{(r)}}(\text{ }\overbrace{j,...,j}^{N}\mid\rho
_{j})\nonumber
\end{align}
where
\begin{equation}
\mu_{\overset{\leadsto}{\mathcal{M}}_{A|\leadsto1\leadsto...\leadsto k}^{(r)}%
}(j_{1},...,j_{k}|\rho_{j})=\mathrm{tr}\left\{  \overset{\leadsto
}{\mathcal{M}}_{A|\leadsto1\leadsto...\leadsto k}^{(r)}(j_{1},...,j_{k}%
)[\rho_{j}]\right\}  . \label{48.1}%
\end{equation}

The posterior state $\overset{\leadsto}{\sigma}_{out}^{(k)}%
(\overbrace{j,...,j}^{k})|\rho_{j})$ conditioned by an outcome
$(\overbrace{j,...,j}^{k})\in\{1,...,r\}^{k}$ observed under a consecutive
measurement (\ref{47.1}) on a state $\rho_{j}$ is given by
\begin{equation}
\overset{\leadsto}{\sigma}_{out}^{(k)}(\overbrace{j,...,j}^{k})|\rho
_{j})=\frac{\overset{\leadsto}{\mathcal{M}}_{A|\leadsto1\leadsto...\leadsto
k}^{(r)}(\text{ }\overbrace{j,...,j}^{k})[\rho_{j}]}{\mu_{\overset{\leadsto
}{\mathcal{M}}_{A|\leadsto1\leadsto...\leadsto k}^{(r)}}(\underbrace{\text{
}j,...,j}_{k}\mid\rho_{j})}. \label{49}%
\end{equation}
As we note after Eq. (\ref{28}), for simplicity, we further omit the upper
decoration at the outcome $($ $\overbrace{j,...,j}^{k}).$

The relation between a state instrument $\mathcal{M}_{k}^{(r)}(j)$ and the POV
measure $\mathrm{M}_{k}^{(r)}(j)$, corresponding to this instrument and given
by (\ref{4_1}) and relation (\ref{49}) imply that, similarly to (\ref{30}) and
(\ref{31}), under an $N$-sequential conclusive state discrimination via a
noisy communication, the success probability (\ref{48}) admits also two other
equivalent representations:%
\begin{align}
&  \mathbb{P}_{\overset{\leadsto}{\mathcal{M}}_{A|\rightsquigarrow
1\leadsto...\leadsto N}^{(r)}}^{success}\left(  \text{ }\rho_{1},...,\rho
_{r}|q_{1},...,q_{r}\right)  |_{\Lambda_{1},...,\Lambda_{_{N}}}\label{49.1}\\
&  =\sum_{j=1,...,r}q_{j}\mathrm{tr}\left\{  \Lambda_{1}[\rho_{j}%
]\mathrm{M}_{1}^{(r)}(j)\right\}  \mathrm{tr}\left\{  \Lambda_{2}\left[
\overset{\leadsto}{\sigma}_{out}^{(1)}(j|\rho_{j})\right]  \text{ }%
\mathrm{M}_{2}^{(r)}(j)\right\} \nonumber\\
&  \ \ \ \ \ \ \ \ \ \ \ \ \ \ \ \ \ \ \times\cdots\times\mathrm{tr}\left\{
\Lambda_{N}\left[  \overset{\leadsto}{\sigma}_{out}^{(N-1)}(j,...,j|\rho
_{j})\right]  \mathrm{M}_{N}^{(r)}(j)\right\}  ,\nonumber
\end{align}
and
\begin{align}
&  \mathbb{P}_{\overset{\leadsto}{\mathcal{M}}_{A|\rightsquigarrow
1\leadsto...\leadsto N}^{(r)}}^{success}\left(  \text{ }\rho_{1},...,\rho
_{r}|q_{1},...,q_{r}\right)  |_{\Lambda_{1},...,\Lambda_{_{N}}}=\mathbb{P}%
_{\mathcal{M}_{1}^{(r)}}^{success}\left(  \text{ }\Lambda_{1}[\rho
_{1}],...,\Lambda_{1}[\rho_{r}]\mid q_{1},...,q_{r}\right) \label{49.2}\\
&  \times%
{\displaystyle\prod\limits_{n=2,...,N}}
\mathbb{P}_{\mathcal{M}_{n}^{(r)}}^{success}\left(  \Lambda_{n}\left[
\overset{\leadsto}{\sigma}_{out}^{(n-1)}(1,...,1|\rho_{1})\right]
,...,\Lambda_{n}\left[  \overset{\leadsto}{\sigma}_{out}^{(n-1)}%
(r,...,r|\rho_{r})\right]  \mid Q_{n-1}^{(1)},...,Q_{n-1}^{(r)}\right)
,\nonumber
\end{align}
\textbf{ }where
\begin{equation}
Q_{n-1}^{(j)}:=\frac{q_{j}\mu_{\overset{\leadsto}{\mathcal{M}}%
_{A|\rightsquigarrow1\leadsto...\leadsto n-1}^{(r)}}\left(  \text{
}j,...,j\mid\rho_{j}\right)  }{\sum_{j}q_{j}\mu_{\overset{\leadsto
}{\mathcal{M}}_{A|\rightsquigarrow1\leadsto...\leadsto n-1}^{(r)}}\left(
\text{ }j,...,j\mid\rho_{j}\right)  } \label{50}%
\end{equation}
are a priori probability of the conditional posterior state originated from
$\rho_{j}$ updated before a measurement of $n$-th receiver.

Also, in case of a noisy communication, the upper bound, similar to
(\ref{36}), takes the form
\begin{equation}
\mathbb{P}_{\overset{\leadsto}{\mathcal{M}}_{A|\rightsquigarrow1\leadsto
...\leadsto N}^{(r)}}^{success}\left(  \text{ }\rho_{1},...,\rho_{r}%
|q_{1},...,q_{r}\right)  |_{\Lambda_{1},...,\Lambda_{_{N}}}\leq\frac{1}%
{r}\left(  1+%
{\displaystyle\sum\limits_{1\leq i<j\leq r}}
\left\Vert \text{ }\Lambda_{1}\left[  q_{i}\rho_{i}-q_{j}\rho_{j}\right]
\right\Vert _{1}\right)  . \label{51}%
\end{equation}

Specified for a two-sequential discrimination between two arbitrary quantum
states $(N=2)$, representation (\ref{49.2}) for the success probability
reduces to%
\begin{align}
&  \mathbb{P}_{\overset{\leadsto}{\mathcal{M}}_{A|\rightsquigarrow1\leadsto
2}^{(r)}}^{success}\left(  \text{ }\rho_{1},...,\rho_{r}|q_{1},...,q_{r}%
\right)  |_{\Lambda_{1,}\Lambda_{2}}=\mathbb{P}_{\mathcal{M}_{1}^{(r)}%
}^{success}\left(  \Lambda_{1}[\rho_{1}],...,\Lambda_{1}[\rho_{r}]\mid
q_{1},...,q_{r}\text{ }\right) \label{52}\\
&  \times\mathbb{P}_{\mathcal{M}_{2}^{(r)}}^{success}\left(  \text{ }%
\Lambda_{2}\left[  \overset{\leadsto}{\sigma}_{out}^{(1)}(1|\rho_{1})\right]
,...,\Lambda_{2}\left[  \overset{\leadsto}{\sigma}_{out}^{(1)}(r|\rho
_{r})\right]  \mid Q_{1}^{(1)},...,Q_{1}^{(r)}\right)  ,\nonumber
\end{align}
where states $\overset{\leadsto}{\sigma}_{out}^{(1)}(j|\rho_{j})$ and updated
a priori probabilities $Q_{1}^{(j)}$ before a measurement of the second
receiver are given by%
\begin{align}
\overset{\leadsto}{\sigma}_{out}^{(1)}(j|\rho_{j})  &  =\frac{\mathcal{M}%
_{1}^{(2)}(j)\left[  \Lambda_{1}[\rho_{j}]\right]  }{\mu_{\mathcal{M}%
_{1}^{(2)}}(j|\Lambda_{1}[\rho_{j}])},\label{52.0}\\
Q_{1}^{(j)}  &  =\frac{q_{j}\mu_{\mathcal{M}_{1}^{(2)}}(j|\Lambda_{1}[\rho
_{j}])}{\mathbb{P}_{\mathcal{M}_{1}^{(2)}}^{success}\left(  \text{ }%
\Lambda_{1}\left[  \rho_{1}\right]  ,...,\Lambda_{1}\left[  \rho_{r}\right]
\mid q_{1},...,q_{r}\right)  },\nonumber
\end{align}
and
\begin{align}
&  \mathbb{P}_{\mathcal{M}_{2}^{(2)}}^{success}\left(  \Lambda_{2}\left[
\overset{\leadsto}{\sigma}_{out}^{(1)}(1|\rho_{1})\right]  ,\ldots,\Lambda
_{2}\left[  \overset{\leadsto}{\sigma}_{out}^{(1)}(r|\rho_{r})\right]  \mid
Q_{1}^{(1)},...,Q_{1}^{(r)}\right) \nonumber\\
&  =\sum_{j=1,...,r}Q_{1}^{(j)}\mathrm{tr}\left\{  \Lambda_{2}\left[
\overset{\leadsto}{\sigma}_{out}^{(1)}(j|\rho_{j})\right]  \text{ }%
\mathrm{M}_{2}^{(2)}(j)\right\}  . \label{52.1}%
\end{align}
From (\ref{52}) it follows that, for\emph{ a two-sequential discrimination}
between arbitrary $r\geq2$ \ quantum states, the optimal success probability
has the form%
\begin{align}
&  \mathbb{P}_{A|\leadsto1\leadsto2}^{opt.success}\left(  \rho_{1}%
,...,\rho_{2}|q_{1},...,q_{r}\right)  |_{\Lambda_{1},\Lambda_{2}}%
=\max_{\mathcal{M}_{1}^{(r)}}{\LARGE (}\mathbb{P}_{\mathcal{M}_{1}^{(r)}%
}^{success}\left(  \text{ }\Lambda_{1}[\rho_{1}],...,\Lambda_{1}[\rho_{r}]\mid
q_{1},...,q_{r}\right)  \text{ }\label{53}\\
&  \ \ \ \ \ \ \ \ \ \ \ \ \ \ \ \times\max_{\mathcal{M}_{2}^{(2)}}%
\mathbb{P}_{\mathcal{M}_{2}^{(2)}}^{success}\left(  \Lambda_{2}\left[
\overset{\leadsto}{\sigma}_{out}^{(1)}(1|\rho_{1})\right]  ,\ldots,\Lambda
_{2}\left[  \overset{\leadsto}{\sigma}_{out}^{(1)}(r|\rho_{r})\right]  \mid
Q_{1}^{(1)},...,Q_{1}^{(r)}\right)  {\LARGE )}\nonumber
\end{align}
and similarly to our note in Remark 2, is not in general equal to the product
of the optimal success probabilities of the first and the second receiver.

Taking further into account the upper bound (\ref{36}), found in
\cite{loubenets}, and also relations (\ref{52.0}), we derive the following
general statement.

\begin{theorem}
The optimal success probability for a two-sequential discrimination between
$r\geq2$ arbitrary quantum states under a noise admits the upper bound
\begin{align}
&  \mathbb{P}_{A|\leadsto1\leadsto2}^{opt.success}\left(  \text{ }\rho
_{1},...,\rho_{2}|q_{1},...,q_{r}\right)  |_{_{\Lambda_{1},\Lambda_{2}}%
}\label{54_}\\
&  \leq\frac{1}{r}\max_{\mathcal{M}_{1}^{(r)}}{\LARGE (}\mathbb{P}%
_{\mathcal{M}_{1}^{(r)}}^{success}\left(  \text{ }\Lambda_{1}[\rho
_{1}],...,\Lambda_{1}[\rho_{r}]\mid q_{1},...,q_{r}\right) \nonumber\\
&  +%
{\displaystyle\sum\limits_{1\leq i<j\leq r}}
\left\Vert \Lambda_{2}\left[  q_{i}\mathcal{M}_{1}^{(r)}(i)\left[  \Lambda
_{1}[\rho_{i}]\right]  -q_{j}\mathcal{M}_{1}^{(r)}(j)\left[  \Lambda_{1}%
[\rho_{j}]\right]  \right]  \right\Vert _{1}{\LARGE )},\nonumber
\end{align}
where%
\begin{equation}
\mathbb{P}_{\mathcal{M}_{1}^{(r)}}^{success}\left(  \text{ }\Lambda_{1}%
[\rho_{1}],...,\Lambda_{1}[\rho_{r}]\mid q_{1},...,q_{r}\right)  \leq\frac
{1}{r}\left(  1+%
{\displaystyle\sum\limits_{1\leq i<j\leq r}}
\left\Vert \Lambda_{1}[q_{i}\rho_{i}-q_{j}\rho_{j}]\right\Vert _{1}\right)  .
\label{54.1}%
\end{equation}
The equality in the first line of (\ref{54_}) holds for $r=2$ and reads:%
\begin{align}
&  \mathbb{P}_{A|\leadsto1\leadsto2}^{opt.success}\left(  \text{ }\rho
_{1},\rho_{2}|q_{1},q_{2}\right)  |_{\Lambda_{1},\Lambda_{2}}\label{55_}\\
&  =\max_{\mathcal{M}_{1}^{(2)}}\frac{1}{2}{\Large (}\mathrm{P}_{\mathcal{M}%
_{1}^{(2)}}^{success}\left(  \text{ }\Lambda_{1}[\rho_{1}],\Lambda_{1}%
[\rho_{2}]\mid q_{1},q_{2}\right) \nonumber\\
&  +\left\Vert \Lambda_{2}\left[  q_{1}\mathcal{M}_{1}^{(2)}(1)\left[
\Lambda_{1}[\rho_{1}]\right]  -q_{2}\mathcal{M}_{1}^{(2)}(2)\left[
\Lambda_{1}[\rho_{2}]\right]  \right]  \right\Vert _{1}{\LARGE ).}\nonumber
\end{align}

\end{theorem}

Since (\ref{55_}) constitutes the maximization of a convex function on a
convex set $\{\mathcal{M}_{1}^{(2)}\}$ of all quantum instruments of the first
receiver, the extremum in (\ref{55_}) is attained at the extreme points of
this set, that is, on the subset of quantum instruments with elements
$P(j)[\cdot]P(j),$ $j=1,2$, where $\mathrm{P}(j)$ are orthogonal projections
on $\mathcal{H}$ and $\sum_{j=1,2}\mathrm{P}(j)=\mathbb{I}_{\mathcal{H}}$.

Taking into account that, for any positive trace class operators
$T,\widetilde{T},$%
\begin{align}
||T-\widetilde{T}||_{1}  &  \leq||T||_{1}+||\widetilde{T}_{2}||_{1}%
\label{55.1}\\
&  =\mathrm{tr}[T]+\mathrm{tr}[\widetilde{T}],\nonumber
\end{align}
and that $\mathrm{tr}\left\{  \Lambda\left[  T\right]  \right\}
=\mathrm{tr}\left\{  T\right\}  $ for any trace class $T,$ we have in
(\ref{54_}):%
\begin{align}
&
{\displaystyle\sum\limits_{1\leq i<j\leq r}}
\left\Vert \Lambda_{2}\left[  q_{i}\mathcal{M}_{1}^{(r)}(i)[\Lambda_{1}%
[\rho_{i}]]-q_{j}\mathcal{M}_{1}^{(r)}(j)[\Lambda_{1}[\rho_{j}]]\right]
\right\Vert _{1}\label{56}\\
&  \leq%
{\displaystyle\sum\limits_{1\leq i<j\leq r}}
\left(  q_{i}\mathrm{tr}\left\{  \Lambda_{2}\left[  \mathcal{M}_{1}%
^{(r)}(i)\left[  \Lambda_{1}[\rho_{i}]\right]  \right]  \right\}
+q_{j}\mathrm{tr}\left\{  \Lambda_{2}\left[  \mathcal{M}_{1}^{(r)}%
(i)[\Lambda_{1}[\rho_{j}]]\right]  \right\}  \right) \nonumber\\
&  =%
{\displaystyle\sum\limits_{1\leq i<j\leq r}}
\left(  q_{i}\mathrm{tr}\left\{  \Lambda_{1}[\rho_{i}]\text{ }\mathrm{M}%
_{1}^{(r)}(i)\right\}  +q_{j}\mathrm{tr}\left\{  \Lambda_{1}[\rho_{j}]\text{
}\mathrm{M}_{1}^{(r)}(j)\right\}  \right) \nonumber\\
&  =(r-1)\text{ }\mathbb{P}_{\mathcal{M}_{1}^{(r)}}^{success}\left(
\Lambda_{1}[\rho_{1}],...,\Lambda_{1}[\rho_{r}]\mid q_{1},...,q_{r}\right)
,\nonumber
\end{align}
which is consistent with the upper bound (\ref{36}) if\textbf{ }$N=2.$

\section{Two-sequential conclusive discrimination under a depolarizing noise}

Consider the optimal success probability (\ref{55_}) for a two-sequential
conclusive discrimination in the case where a sender prepares two qubit states
$\rho_{j},$ \ $j=1,2,$ with Bloch representations
\begin{equation}
\rho_{j}=\frac{1}{2}(\mathbb{I}_{2}+\vec{r}_{j}\cdot\vec{\sigma}),\text{
\ \ }\vec{r}_{j}=\mathrm{tr}\{\rho_{j}\vec{\sigma}\}\in\mathbb{R}^{3},
\label{57}%
\end{equation}
and a priori probabilities $q_{j}$, $j=1,2,$ and quantum communication
channels between a sender and a receiver and between two receives are
depolarizing%
\begin{equation}
\Lambda_{1}^{(pol)}[T]:=(1-\gamma_{1})T+\gamma_{1}\mathrm{tr}\{T\}\text{
}\frac{\mathbb{I}_{2}}{2},\ \text{\ \ }\gamma_{1}\in\lbrack0,1],\text{
\ }\ n=1,2. \label{58}%
\end{equation}
Here, $\mathbb{I}_{2}$ is an identity operator and $T$ is an arbitrary linear
operator on\ $\mathbb{C}^{2}.$ For a qubit state $\rho_{j}$ with
representation (\ref{57}), relation (\ref{58}) reads
\begin{equation}
\Lambda_{1}^{(pol)}[\rho_{j}]=\frac{1}{2}\left(  \mathbb{I}_{2}+(1-\gamma
_{1})\text{ }\vec{r}_{j}\cdot\vec{\sigma}\right)  . \label{59}%
\end{equation}

For only one receiver $(N=1)$, the optimal success probability under a
depolarizing noise is given by the Helstrom bound (\ref{37}) for the
discrimination between noisy states $\Lambda_{1}^{(pol)}[\rho_{1}]$ and
$\Lambda_{1}^{(pol)}[\rho_{2}]$, that is:%
\begin{equation}
\mathbb{P}_{A|\leadsto1}^{opt.success}(\rho_{1},\rho_{2}|q_{1},q_{2}%
)|_{_{_{\gamma_{1}}}}:=\frac{1}{2}\left(  1+\left\Vert q_{1}\Lambda
_{1}^{(pol)}[\rho_{1}]-q_{2}\Lambda_{1}^{(pol)}[\rho_{2}]\right\Vert
_{1}\right)  \label{60}%
\end{equation}
Taking into account that, for a Hermitian operator on a finite-dimensional
Hilbert space, the trace norm is given by the sum of the absolute values of
its eigenvalues and that the eigenvalues of the qubit operator $\left(
q_{1}\Lambda_{1}^{(pol)}[\rho_{1}]-q_{2}\Lambda_{1}^{(pol)}[\rho_{2}]\right)
$ are equal to $\frac{1}{2}\left(  q_{1}-q_{2}\pm(1-\gamma_{1})\left\Vert
q_{1}\vec{r}_{1}-q_{2}\vec{r}_{2}\right\Vert _{\mathbb{R}^{3}}\right)  ,$ we
derive
\begin{align}
\mathbb{P}_{A|\leadsto1}^{opt.success}(\rho_{1},\rho_{2}|q_{1},q_{2}%
)|_{_{_{\gamma_{1}}}}  &  =\frac{1}{2}+\frac{1}{4}\left\vert \text{ }%
q_{1}-q_{2}+(1-\gamma_{1})\left\Vert q_{1}\vec{r}_{1}-q_{2}\vec{r}%
_{2}\right\Vert _{\mathbb{R}^{3}}\right\vert \label{60.1}\\
&  +\frac{1}{4}\left\vert \text{ }q_{1}-q_{2}-(1-\gamma_{1})\left\Vert
q_{1}\vec{r}_{1}-q_{2}\vec{r}_{2}\right\Vert _{\mathbb{R}^{3}}\right\vert
.\nonumber
\end{align}
From (\ref{60}) and (\ref{60.1}) it follows that%
\begin{equation}
\mathbb{P}_{A|\leadsto1}^{opt.success}(\rho_{1},\rho_{2}|q_{1},q_{2}%
)|_{_{_{\gamma_{1}}}}=\frac{1}{2}+\frac{1}{2}(1-\gamma_{1})\left\Vert
q_{1}\vec{r}_{1}-q_{2}\vec{r}_{2}\right\Vert _{\mathbb{R}^{3}} \label{61}%
\end{equation}
if
\begin{align}
(1-\gamma_{1})\left\Vert q_{1}\vec{r}_{1}-q_{2}\vec{r}_{2}\right\Vert
_{\mathbb{R}^{3}}  &  >\left\vert q_{1}-q_{2}\right\vert \text{ \ \ }%
\Leftrightarrow\label{62}\\
\frac{1}{2}+\frac{1}{2}(1-\gamma_{1})\left\Vert q_{1}\vec{r}_{1}-q_{2}\vec
{r}_{2}\right\Vert _{\mathbb{R}^{3}}  &  \geq\max\{q_{1},q_{2}\},\nonumber
\end{align}
and
\begin{equation}
\mathbb{P}_{A|\leadsto1}^{opt.success}(\rho_{1},\rho_{2}|q_{1},q_{2}%
)|_{_{_{\gamma_{1}}}}=\max\{q_{1},q_{2}\} \label{63}%
\end{equation}
otherwise.

Relations (\ref{60}) --(\ref{63}) imply
\begin{equation}
\mathbb{P}_{A|\leadsto1}^{opt.success}(\rho_{1},\rho_{2}|q_{1},q_{2}%
)|_{_{_{\gamma_{1}}}}=\max\left\{  \frac{1}{2}+\frac{1}{2}(1-\gamma
_{1})\left\Vert q_{1}\vec{r}_{1}-q_{2}\vec{r}_{2}\right\Vert _{\mathbb{R}^{3}%
},\max\{q_{1},q_{2}\}\right\}  . \label{64}%
\end{equation}

\begin{remark}
If initial qubit states are pure: $\rho_{1}=|\psi_{1}\rangle\langle\psi_{1}|,$
$\rho_{2}=|\psi_{2}\rangle\langle\psi_{2}|$, then $\left\Vert \vec{r}%
_{1}\right\Vert _{\mathbb{R}^{3}}=\left\Vert \vec{r}_{2}\right\Vert
_{\mathbb{R}^{3}}=1$ and, as well-known,
\begin{align}
\left\Vert q_{1}\vec{r}_{1}-q_{2}\vec{r}_{2}\right\Vert _{\mathbb{R}^{3}}  &
=\sqrt{q_{1}^{2}+q_{2}^{2}-2q_{1}q_{2}\vec{r}_{1}\cdot\vec{r}_{2}}%
\label{64_1}\\
&  =\sqrt{1-4q_{1}q_{2}\left\vert \langle\psi_{2}|\psi_{1}\rangle\right\vert
^{2}}.\nonumber
\end{align}
Therefore, for pure states, expression (\ref{64}) reduces to%
\begin{align}
&  \mathbb{P}_{A|\leadsto1}^{opt.success}(\rho_{1},\rho_{2}|q_{1}%
,q_{2})|_{_{_{\gamma_{1}}}}\label{64_2}\\
&  =\max\left\{  \frac{1}{2}+\frac{1}{2}(1-\gamma_{1})\sqrt{1-4q_{1}%
q_{2}\left\vert \langle\psi_{2}|\psi_{1}\rangle\right\vert ^{2}},\text{ }%
\max\{q_{1},q_{2}\}\right\}  .\nonumber
\end{align}
Note that since%
\begin{equation}
\left\Vert q_{1}\vec{r}_{1}-q_{2}\vec{r}_{2}\right\Vert _{\mathbb{R}^{3}%
}=\sqrt{1-4q_{1}q_{2}\left\vert \langle\psi_{2}|\psi_{1}\rangle\right\vert
^{2}}\geq\left\vert q_{1}-q_{2}\right\vert , \label{64_3}%
\end{equation}
for $\gamma_{1}=0,$ maximum (\ref{64_2}) reduces the well-known expression
\begin{equation}
\frac{1}{2}+\frac{1}{2}\sqrt{1-4q_{1}q_{2}\left\vert \langle\psi_{2}|\psi
_{1}\rangle\right\vert ^{2}} \label{pure_hel}%
\end{equation}
for the Helstrom bound in case of pure initial states.
\end{remark}

For two receivers $(N=2)$, we derive the following result proved rigorously in Appendix.

\begin{theorem}
Let $\rho_{1}$ and $\rho_{2}$ be arbitrary qubit states given with a priori
probabilities $q_{1}$, $q_{2}$. In case of depolarizing channels (\ref{58}),
the optimal success probability (\ref{55_}) of a two-sequential discrimination
between qubit states $\rho_{1}$ and $\rho_{2}$ is given by%
\begin{align}
&  \mathbb{P}_{A|\leadsto1\leadsto2}^{opt.success}\left(  \rho_{1},\rho
_{2}|q_{1},q_{2}\right)  |_{_{\gamma_{_{1}},\gamma_{_{2}}}}\label{65}\\
&  =\max\left\{  \left(  1-\frac{\gamma_{2}}{2}\right)  \left(  \frac{1}%
{2}+\frac{1}{2}(1-\gamma_{1})\left\Vert q_{1}\vec{r}_{1}-q_{2}\vec{r}%
_{2}\right\Vert _{\mathbb{R}^{3}}\right)  ,\text{ }\max\{q_{1},q_{2}%
\}\right\}  ,\nonumber
\end{align}
where $\vec{r}_{j},$ $j=1,2$ are Bloch vectors of states $\rho_{j}$ in
representation (\ref{57}).
\end{theorem}

For an arbitrary $\gamma_{1}\in\lbrack0,1]$ and $\gamma_{2}=0,$ expression
(\ref{65}) coincides with expression (\ref{64}) and this agrees with
Proposition 2.

In Fig.1 and Fig. 2, we present the numerical results related to Eqs.
(\ref{64}) and \ref{65}. For simplicity, we assume that $\gamma_{1}=\gamma
_{2}=\gamma.$

In Fig.1, we take qubit states $\rho_{1}$ and $\rho_{2}$ with the Bloch
vectors
\begin{equation}
\vec{r}_{1}=(0.3,0.3,0.3),\ \ \ \vec{r}_{2}=(0.3,0.3,-0.3). \label{ex1}%
\end{equation}
In Fig.1(b) we consider the case with equal a priori probabilities
$q_{1}=q_{2}=0.5$ and in Fig. 1(b) with the a priori probabilities
$q_{1}=0.55,$ $q_{2}=0.45$.

In Fig.2, we take qubit states $\rho_{1}$ and $\rho_{2}$ with the Bloch
vectors
\begin{equation}
\vec{r}_{1}=(0.2,0.3,-0.4),\ \ \ \vec{r}_{2}=(-0.2,-0.3,0.35), \label{ex2}%
\end{equation}
and with the a priori probabilities $q_{1}=q_{2}=0.5$ in Fig.2(a) and
$q_{1}=0.55,$ $q_{2}=0.45$ in Fig.2(b).

In Figs.1-2: (i) the solid red line is the Helstrom bound (\ref{37}) for
states $\rho_{1}$ and $\rho_{2}$ and the corresponding a priori probabilities
$q_{1},q_{2}$; (ii) the solid black line is the optimal success probability
(\ref{64}) for one receiver ($N=1$) under a depolarizing channel between a
sender and a receiver; (iii) the solid blue line is the optimal two-sequential
success probability (\ref{65}) under depolarizing quantum channels between a
sender and a receiver and between two receivers.

According to the presented numerical results, if $\gamma$ is not so large, the
dependence of the optimal success probability (\ref{65}) in $\gamma$ is
presented by a nearly straight line. This is since, for the considered states
and a priori probabilities, in the corresponding regions of $\gamma,$ the
impact in (\ref{65}) of the term quadratic in $\gamma$ is much smaller than
that of the term linear in $\gamma$.

\section{Conclusion}

In the present article, we have developed a general framework for the
description of an $N$-sequential conclusive discrimination between any number
$r\geq2$ of arbitrary quantum states, pure or mixed. For this general
sequential discrimination scenario, we have:

\begin{itemize}
\item derived (Proposition 1)three mutually equivalent general representations
(\ref{25}), (\ref{30}), (\ref{31}) for the success probability;

\item found (Theorem 1) a new general condition on $r\geq2$ quantum states
sufficient for the optimal success probability to be equal to the optimal
success probability of the first receiver for any number $N\geq2$ of further
sequential receivers and specified the corresponding optimal protocol
(\ref{44.2});

\item shown (Theorem 2) that, in case of discrimination between two arbitrary
quantum states, this sufficient condition is always fulfilled, so that the
optimal success probability of an $N$-sequential conclusive discrimination
between two arbitrary quantum states is given by the Helstrom bound for any
number $N\geq2$ of sequential receivers and is attained under the optimal
protocols (\ref{39}) and (\ref{39_1}). Each of these optimal protocols is
general in the sense that it is true for the $N$-sequential conclusive
sequential discrimination between any two quantum states, pure or mixed, and
of an arbitrary dimension.

\item explicitly constructed (Proposition 2) receivers' indirect measurements
implementing the general optimal protocol (\ref{39}) for the $N$-sequential
conclusive discrimination between any two quantum states.
\end{itemize}

Furthermore, we have extended our general framework to include by
(\ref{49.1}), (\ref{49.2}) the case of an $N$-sequential conclusive
discrimination between $r\geq2$ arbitrary quantum states under arbitrary noisy
quantum communication channels. For the optimal success probability in case of
a two-sequential discrimination under a noise between $r\geq2$ arbitrary
quantum states, we have specified (Theorem 3) a new general upper bound
(\ref{54_}), which is attained if $r=2$.

The developed general framework is true for any number $N\geq1$ of sequential
receivers, any number $r\geq2$ of arbitrary quantum states, pure or mixed, to
be discriminated, all types of receivers' quantum measurements and arbitrary
noisy quantum communication channels.

As an example, we analyze analytically (Theorem 4) and further numerically a
two-sequential conclusive discrimination between two qubit states via
depolarizing quantum channels.

The new general results derived within the developed framework are important
both from the theoretical point of view and for a successful multipartite
quantum communication via noisy channels.

\section*{Acknowledgement}

The study by E. R. Loubenets in Section 2, Section 3 and in Section 4.1 of
this work was supported by the Russian Science Foundation under the Grant No
19-11-00086 and performed at the Steklov Mathematical Institute of Russian
Academy of Sciences. The study by E. R. Loubenets in Section 5 and in Section
6 was performed at the National Research University Higher School of
Economics. The study by Min Namkung in Section 4.2, Section 5 and in Section 6
was performed at the National Research University Higher School of Economics
and during the revision -- at the Kyung Hee University. Min Namkung
acknowledges support from the National Research Foundation of Korea (NRF)
grant (NRF2020M3E4A1080088) funded by the Korea government (Ministry of
Science and ICT).

The authors acknowledge the valuable comments of anonymous Reviewers. Also,
Min Namkung thanks Prof. Younghun Kwon at Hanyang University (ERICA) for his
insightful discussion.

\section{Appendix}

In this Section, we present the proof of Theorem 4. In the qubit case, the
extreme points of set $\{\mathcal{M}_{1}^{(2)}\}$ in (\ref{55_}) are given by
(i) the state instruments $\left\{  \Pi_{\vec{n}}^{(\pm)}(),\text{ }\vec{n}%
\in\mathbb{R}^{3},\text{ }\left\Vert \vec{n}\right\Vert _{\mathbb{R}^{3}%
}=1\right\}  \ $with elements%
\begin{align}
\Pi_{\vec{n}}^{(\pm)}(1)[\cdot]  &  =\mathrm{E}_{\sigma_{\vec{n}}}^{(\pm
)}\text{ }[\cdot]\text{ }\mathrm{E}_{\sigma_{\vec{n}}}^{(\pm)},\text{
\ \ \ }\Pi_{\vec{n}}^{(\pm)}(2)[\cdot]=\mathrm{E}_{\sigma_{\vec{n}}}^{(\mp
)}\text{ }[\cdot]\text{ }\mathrm{E}_{\sigma_{\vec{n}}}^{(\mp)}, \tag{A1}%
\label{A1}\\
\mathrm{E}_{\sigma_{\vec{n}}}^{(\pm)}  &  =\frac{\mathbb{I}\text{ }%
\mathbb{\pm}\text{ }\sigma_{\vec{n}}}{2},\nonumber
\end{align}
where $\mathrm{E}_{\sigma_{\vec{n}}}^{(\pm)}$ are the spectral projections
corresponding to eigenvalues $\pm1$ of a Hermitian operator $\sigma_{\vec{n}%
}:=\vec{n}\cdot\vec{\sigma}$ describing the projection of qubit spin
$\vec{\sigma}$ on a unit direction $\vec{n}$ and (ii) the state instruments
$\Pi_{j}(\cdot),$ $j=1,2,$ with elements
\begin{align}
\Pi_{1}(1)[\cdot]  &  =\mathbb{I}_{2}\text{ }[\cdot]\text{ }\mathbb{I}%
_{2},\text{ \ \ \ }\Pi_{1}(2)[\cdot]=0,\tag{A2}\label{A2}\\
\Pi_{2}(1)[\cdot]  &  =0,\text{ \ \ \ }\Pi_{2}(2)=\mathbb{I}_{2}\text{ }%
[\cdot]\text{ }\mathbb{I}_{2}.\nonumber
\end{align}
Therefore,
\begin{align}
\mathbb{P}_{A|\leadsto1\leadsto2}^{opt.success}\left(  \rho_{1},\rho_{2}%
|q_{1},q_{2}\right)  |_{_{\gamma_{_{1}},\gamma_{_{2}}}}  &  =\max_{\Pi\in
\{\Pi_{\vec{n}}^{(\pm)},\Pi_{j}\}}\mathcal{E}_{\Pi}(\rho_{1},\rho_{2}%
,q_{1},q_{2})|_{_{\gamma_{1},\gamma_{2}}},\tag{A3}\label{A3}\\
&  =\max\left\{  \max_{\vec{n}}\mathcal{E}_{\Pi_{\vec{n}}^{(\pm)}}(\rho
_{1},\rho_{2},q_{1},q_{2})|_{_{\gamma_{1},\gamma_{2}}},\max_{j-1,2}%
\mathcal{E}_{\Pi_{j}}(\rho_{1},\rho_{2},q_{1},q_{2})\right\}  ,\nonumber
\end{align}
where $\mathcal{E}_{\Pi}(\rho_{1},\rho_{2},q_{1},q_{2})|_{_{\gamma_{1}%
,\gamma_{2}}}$\textrm{ }is the expression to be maximized in\textrm{
}(\ref{55_}).

From Eqs. (\ref{58}), (\ref{59}) and (\ref{A1}) it follows
\begin{equation}
q_{j}\mathrm{tr}\left\{  \Lambda_{1}[\rho_{j}]\mathrm{E}_{\sigma_{\vec{n}}%
}^{(\pm)}\right\}  =\frac{1}{2}q_{j}\pm\frac{1}{2}q_{j}(1-\gamma_{1})(\vec
{r}_{j}\cdot\vec{n}), \tag{A4}\label{A4}%
\end{equation}
so that, in expression $\mathcal{E}_{\Pi_{\vec{n}}^{(\pm)}}$, standing under
maximum in (\ref{55_}),\emph{ for each extreme state instrument} $\Pi_{\vec
{n}}^{(\pm)}$, the first term equals to
\begin{align}
\mathbb{P}_{\Pi_{\vec{n}}^{(\pm)}}^{success}\left(  \text{ }\rho_{1},\rho
_{2}\mid q_{1},q_{2}\right)  |_{_{\gamma_{1}}}  &  :=\mathbb{P}_{\Pi_{\vec{n}%
}^{(\pm)}}^{success}\left(  \text{ }\Lambda_{1}[\rho_{1}],\Lambda_{1}[\rho
_{2}]\mid q_{1},q_{2}\right)  =\tag{A5}\label{A5}\\
&  =q_{1}\mathrm{tr}\left\{  \Lambda_{1}[\rho_{1}]\text{ }\mathrm{E}%
_{\sigma_{\vec{n}}}^{(\pm)}\right\}  +q_{2}\mathrm{tr}\left\{  \Lambda
_{1}[\rho_{2}]\text{ }\mathrm{E}_{\sigma_{\vec{n}}}^{(\mp)}\right\}
\nonumber\\
&  =\frac{1}{2}\pm\frac{1}{2}\left(  \text{ }1-\gamma_{1}\right)  \left(
\text{ }q_{1}\vec{r}_{1}-q_{2}\vec{r}_{2}\right)  \cdot\vec{n},\nonumber
\end{align}
and the second term -- \emph{the trace norm} of the Hermitian operator%
\begin{align}
&  \Lambda_{2}\left[  q_{1}\Pi_{\vec{n}}^{(\pm)}(1)\left[  \Lambda_{1}%
[\rho_{1}]\right]  -q_{2}\Pi_{\vec{n}}^{(\pm)}(2)\left[  \Lambda_{1}[\rho
_{2}]\right]  \right] \tag{A6}\label{A6}\\
&  =\gamma_{2}\frac{\mathbb{I}}{2}\left[  \text{ }q_{1}\mathrm{tr}\left\{
\Lambda_{1}[\rho_{1}]\mathrm{E}_{\sigma_{\vec{n}}}^{(\pm)}\right\}
-q_{2}\mathrm{tr}\left\{  \Lambda_{1}[\rho_{2}]\mathrm{E}_{\sigma_{\vec{n}}%
}^{(\mp)}\right\}  \right] \nonumber\\
&  +(1-\gamma_{2})\left[  q_{1}\mathrm{E}_{\sigma_{\vec{n}}}^{(\pm)}\text{
}\mathrm{tr}\left\{  \Lambda_{1}[\rho_{1}]\mathrm{E}_{\sigma_{\vec{n}}}%
^{(\pm)}\right\}  -q_{2}\mathrm{E}_{\sigma_{\vec{n}}}^{(\mp)}\text{
}\mathrm{tr}\left\{  \Lambda_{1}[\rho_{2}]\mathrm{E}_{\sigma_{\vec{n}}}%
^{(\mp)}\right\}  \right] \nonumber\\
&  =\mathrm{E}_{\sigma_{\vec{n}}}^{(\pm)}\left\{  \left(  1-\frac{\gamma_{2}%
}{2}\right)  \mathbb{P}_{\Pi_{\vec{n}}^{(\pm)}}^{success}|_{_{\gamma_{1}}%
}-q_{2}\mathrm{tr}\left\{  \Lambda_{1}[\rho_{2}]\mathrm{E}_{\sigma_{\vec{n}}%
}^{(\mp)}\right\}  \right\} \nonumber\\
&  -\mathrm{E}_{\sigma_{\vec{n}}}^{(\mp)}\left\{  \left(  1-\frac{\gamma_{2}%
}{2}\right)  \mathbb{P}_{\Pi_{\vec{n}}^{(\pm)}}^{success}|_{_{\gamma_{1}}%
}-q_{1}\mathrm{tr}\left\{  \Lambda_{1}[\rho_{1}]\mathrm{E}_{\sigma_{\vec{n}}%
}^{(\pm)}\right\}  \right\}  ,\nonumber
\end{align}
is equal to%
\begin{align}
&  \left\vert \left(  1-\frac{\gamma_{2}}{2}\right)  \mathbb{P}_{\Pi_{\vec{n}%
}^{(\pm)}}^{success}|_{_{\gamma_{1}}}-q_{2}\mathrm{tr}\left\{  \Lambda
_{1}[\rho_{2}]\text{ }\mathrm{E}_{\sigma_{\vec{n}}}^{(\mp)}\right\}
\right\vert \tag{A7}\label{A7}\\
&  +\left\vert \left(  1-\frac{\gamma_{2}}{2}\right)  \mathbb{P}_{\Pi_{\vec
{n}}^{(\pm)}}^{success}|_{_{\gamma_{1}}}-q_{1}\mathrm{tr}\left\{  \Lambda
_{1}[\rho_{1}]\text{ }\mathrm{E}_{\sigma_{\vec{n}}}^{(\pm)}\right\}
\right\vert .\nonumber
\end{align}
Here, the terms $q_{j}\mathrm{tr}\left\{  \Lambda_{1}[\rho_{j}]E_{\sigma
_{\vec{n}}}^{(\pm)}\right\}  $ and $\mathbb{P}_{\Pi_{\vec{n}}^{(\pm)}%
}^{success}(\rho_{1},\rho_{2}|q_{1},q_{2})|_{\gamma_{1}}$ are given by
(\ref{A4}) and (\ref{A5}), respectively, and, in (\ref{A6}), we denote
$\mathbb{P}_{\Pi_{\vec{n}}^{(\pm)}}^{success}|_{_{\gamma_{1}}}:=\mathbb{P}%
_{\Pi_{\vec{n}}^{(\pm)}}^{success}(\rho_{1},\rho_{2}|q_{1},q_{2})|_{\gamma
_{1}}$ for short.

For the extreme state instruments (\ref{A2}), we have%
\begin{align}
\mathbb{P}_{\Pi_{j}}^{success}(\rho_{1},\rho_{2}  &  \mid q_{1},q_{2}%
|_{_{\gamma_{1}}}:=\mathbb{P}_{\Pi_{j}}^{success}\left(  \Lambda_{1}[\rho
_{1}],\Lambda_{1}[\rho_{2}]\mid q_{1},q_{2}\right) \tag{A8}\label{A8}\\
&  =q_{j},\nonumber
\end{align}
also the trace norm standing in (\ref{55_}) is equal to $q_{j}.$

Eqs. (\ref{A5})-(\ref{A8}) imply that, for the state instruments (\ref{A1})
and (\ref{A2}), expression $\mathcal{E}_{\Pi}(\cdot)|_{_{\gamma_{1},\gamma
_{2}}}$\textrm{ }in (\ref{A3}),\textrm{ }takes the following values:%
\begin{align}
\mathcal{E}_{\Pi_{\vec{n}}^{(\pm)}}(\rho_{1},\rho_{2}|q_{1},q_{2}%
)|_{_{\gamma_{1},\gamma_{2}}}  &  =\frac{1}{2}\mathbb{P}_{\Pi_{\vec{n}}%
^{(\pm)}}^{success}(\rho_{1},\rho_{2}|q_{1},q_{2})|_{\gamma_{1}}%
\tag{A9}\label{A9}\\
&  +\frac{1}{2}\left\vert \left(  1-\frac{\gamma_{2}}{2}\right)
\mathbb{P}_{\Pi_{\vec{n}}^{(\pm)}}^{success}(\rho_{1},\rho_{2}\mid q_{1}%
,q_{2})|_{_{\gamma_{1}}}-q_{2}\mathrm{tr}\left\{  \Lambda_{1}[\rho_{2}]\text{
}\mathrm{E}_{\sigma_{\vec{n}}}^{(\mp)}\right\}  \right\vert \nonumber\\
&  +\frac{1}{2}\left\vert \left(  1-\frac{\gamma_{2}}{2}\right)
\mathbb{P}_{\Pi_{\vec{n}}^{(\pm)}}^{success}(\rho_{1},\rho_{2}\mid q_{1}%
,q_{2})|_{_{\gamma_{1}}}-q_{1}\mathrm{tr}\left\{  \Lambda_{1}[\rho_{1}]\text{
}\mathrm{E}_{\sigma_{\vec{n}}}^{(\pm)}\right\}  \right\vert ,\nonumber\\
& \nonumber\\
\mathcal{E}_{\Pi_{j}}(\rho_{1},\rho_{2}|q_{1},q_{2})|_{_{\gamma_{1},\gamma
_{2}}}  &  =q_{j},\text{ \ \ }j=1,2. \tag{A10}\label{A10}%
\end{align}
Substituting (\ref{A9}), (\ref{A10}) into (\ref{A3}), we have%
\begin{equation}
\mathbb{P}_{A|\leadsto1\leadsto2}^{opt.success}\left(  \rho_{1},\rho_{2}%
|q_{1},q_{2}\right)  |_{_{\gamma_{_{1}},\gamma_{_{2}}}}=\max\left\{
\max_{\vec{n}}\mathcal{E}_{\Pi_{\vec{n}}^{(\pm)}}(\rho_{1},\rho_{2}%
,q_{1},q_{2})|_{_{\gamma_{1},\gamma_{2}}},\text{ }\max\mathcal{\{}q_{1}%
,q_{2}\right\}  \tag{A11}\label{A11}%
\end{equation}

Furthermore, relation (\ref{A9}) implies
\begin{equation}
\mathcal{E}_{\Pi_{\vec{n}}^{(\pm)}}(\rho_{1},\rho_{2}|q_{1},q_{2}%
)|_{_{\gamma_{1},\gamma_{2}}}=(1-\frac{\gamma_{2}}{2})\text{ }\mathbb{P}%
_{\Pi_{\vec{n}}^{(\pm)}}^{success}(\rho_{1},\rho_{2}|q_{1},q_{2})|_{\gamma
_{1}} \tag{A12}\label{A12}%
\end{equation}
if%
\begin{equation}
\left(  1-\frac{\gamma_{2}}{2}\right)  \mathbb{P}_{\Pi_{\vec{n}}^{(\pm)}%
}^{success}(\rho_{1},\rho_{2}|q_{1},q_{2})|_{_{\gamma_{1}}}\geq\max\left\{
q_{1}\mathrm{tr}\left\{  \Lambda_{1}[\rho_{1}]\text{ }\mathrm{E}_{\sigma
_{\vec{n}}}^{(\pm)}\right\}  ,q_{2}\mathrm{tr}\left\{  \Lambda_{1}[\rho
_{2}]\text{ }\mathrm{E}_{\sigma_{\vec{n}}}^{(\pm)}\right\}  \right\}  ,
\tag{A13}\label{A13}%
\end{equation}
and
\begin{align}
\mathcal{E}_{\Pi_{\vec{n}}^{(\pm)}}(\rho_{1},\rho_{2}|q_{1},q_{2}%
)|_{_{\gamma_{1},\gamma_{2}}}  &  =\max\left\{  q_{1}\mathrm{tr}\left\{
\Lambda_{1}[\rho_{1}]\text{ }\mathrm{E}_{\sigma_{\vec{n}}}^{(\pm)}\right\}
,q_{2}\mathrm{tr}\left\{  \Lambda_{1}[\rho_{2}]\text{ }\mathrm{E}%
_{\sigma_{\vec{n}}}^{(\mp)}\right\}  \right\} \tag{A14}\label{A14}\\
&  \leq\max\mathcal{\{}q_{1},q_{2}\},\nonumber
\end{align}
otherwise.

Taking relations (\ref{A12}) and (\ref{A14}) into account in expression
(\ref{A11}), we derive%
\begin{equation}
\mathbb{P}_{A|\leadsto1\leadsto2}^{opt.success}\left(  \rho_{1},\rho_{2}%
|q_{1},q_{2}\right)  |_{_{\gamma_{_{1}},\gamma_{_{2}}}}=\max\left\{
\max_{\vec{n}\in J_{1}}\{(1-\frac{\gamma_{2}}{2})\mathbb{P}_{\Pi_{\vec{n}%
}^{(\pm)}}^{success}(\rho_{1},\rho_{2}|q_{1},q_{2})|_{\gamma_{1}}\},\text{
}\max\mathcal{\{}q_{1},q_{2}\}\right\}  \tag{A15}\label{A15}%
\end{equation}
where $J_{1}$ is the set of vectors $\vec{n}\in\mathbb{R}^{3},$ satisfying
relation (\ref{A13}). Note that since by (\ref{A5})%
\begin{equation}
\mathbb{P}_{\Pi_{\vec{n}}^{(\pm)}}^{success}(\rho_{1},\rho_{2}\mid q_{1}%
,q_{2})|_{_{\gamma_{1}}}=\frac{1}{2}\pm\frac{1}{2}(1-\gamma_{1})(q_{1}\vec
{r}_{1}-q_{2}\vec{r}_{2})\cdot\vec{n}, \tag{A16}\label{A16}%
\end{equation}
the maximum of this expression over $\vec{n}\in J_{1}$ is attained at the unit
vectors
\begin{equation}
\vec{n}_{\pm}=\frac{\pm(q_{1}\vec{r}_{1}-q_{2}\vec{r}_{2})}{\left\Vert
q_{1}\text{ }\vec{r}_{1}-q_{2}\text{ }\vec{r}_{2}\right\Vert _{\mathbb{R}^{3}%
}}, \tag{A17}\label{A17}%
\end{equation}
respectively, and is equal to
\begin{align}
\max_{\vec{n}\in J_{1}}\left\{  \mathbb{P}_{\Pi_{\vec{n}}^{(\pm)}}%
^{success}(\rho_{1},\rho_{2}|q_{1},q_{2})|_{\gamma_{1}}\right\}   &
=\mathbb{P}_{\Pi_{\vec{n}_{+}}^{(+)}}^{success}\left(  \rho_{1},\rho_{2}%
|q_{1},q_{2}\right)  |_{_{_{\gamma_{1}}}}=\mathbb{P}_{\Pi_{\vec{n}_{-}}^{(-)}%
}^{success}\left(  \rho_{1},\rho_{2}|q_{1},q_{2}\right)  |_{_{_{\gamma_{1}}}%
}\tag{A18}\label{A18}\\
&  =\frac{1}{2}+\frac{1}{2}(1-\gamma_{1})\left\Vert q_{1}\vec{r}_{1}-q_{2}%
\vec{r}_{2}\right\Vert _{\mathbb{R}^{3}}.\nonumber
\end{align}

From Eqs. (\ref{A15}) and (\ref{A18}) it follows that, under a depolarizing
noise, \emph{the optimal two-sequential success probability }$\mathbb{P}%
_{A|\leadsto1\leadsto2}^{opt.success}\left(  \rho_{1},\rho_{2}|q_{1}%
,q_{2}\right)  |_{_{\gamma_{_{1}},\gamma_{_{2}}}}$

\begin{itemize}
\item equals to%
\begin{align}
&  \max\left\{  \left(  1-\frac{\gamma_{2}}{2}\right)  \mathbb{P}_{\Pi
_{\vec{n}_{+}}^{(+)}}^{success}\left(  \rho_{1},\rho_{2}|q_{1},q_{2}\right)
|_{_{_{\gamma_{1}}}},\max\{q_{1},q_{2}\}\right\} \tag{A19}\label{A19}\\
&  =\max\left\{  \left(  1-\frac{\gamma_{2}}{2}\right)  \left(  \frac{1}%
{2}+\frac{1}{2}(1-\gamma_{1})\left\Vert q_{1}\vec{r}_{1}-q_{2}\vec{r}%
_{2}\right\Vert _{\mathbb{R}^{3}}\right)  ,\text{ }\max\{q_{1},q_{2}\}\right\}
\nonumber
\end{align}
for all $\gamma_{1},\gamma_{2}$ satisfying relation (\ref{A13}) with $\vec
{n}=\vec{n}_{+}$;

\item equals to
\begin{equation}
\max\{q_{1},q_{2}\} \tag{A20}\label{A20}%
\end{equation}
for all otherwise $\gamma_{1},\gamma_{2}.$
\end{itemize}

Note that%
\begin{equation}
\mathbb{P}_{A|\leadsto1}^{opt.success}(\rho_{1},\rho_{2}|q_{1},q_{2}%
)=\max\left\{  \mathbb{P}_{\Pi_{\vec{n}_{+}}^{(+)}}^{success}\left(  \rho
_{1},\rho_{2}|q_{1},q_{2}\right)  |_{_{_{\gamma_{1}}}},\max\{q_{1}%
,q_{2}\}\right\}  \tag{A21}\label{A21}%
\end{equation}
where $\mathrm{P}_{\Pi_{\vec{n}_{+}}^{(+)}}^{success}|_{_{_{\gamma_{1}}}}$ is
given by (\ref{A5}) with $\vec{n}=\vec{n}_{+}$

For a noisy parameter $\gamma_{2},$ let us now specify the range where the
optimal success probability $\mathbb{P}_{A|\leadsto1\leadsto2}^{opt.success}%
|_{_{\gamma_{_{1}},\gamma_{_{2}}}}$ is given by expression (\ref{A19}).

From relation (\ref{A13}), specified for $\vec{n}=\vec{n}_{+}$, it follows
that, \emph{for a given} $\gamma_{1},$ expression (\ref{A19}) holds for all
$\gamma_{2}\leq\gamma_{2}^{(1)}$ where%
\begin{align}
1-\frac{\gamma_{2}^{(1)}}{2}  &  =\frac{\max\left\{  q_{1}\mathrm{tr}\left\{
\Lambda_{1}[\rho_{1}]\text{ }\mathrm{E}_{\sigma_{\vec{n}}}^{(\pm)}\right\}
,q_{2}\mathrm{tr}\left\{  \Lambda_{1}[\rho_{2}]\text{ }\mathrm{E}%
_{\sigma_{\vec{n}}}^{(\mp)}\right\}  \right\}  }{\mathbb{P}_{\Pi_{\vec{n}_{+}%
}^{(+)}}^{success}\left(  \rho_{1},\rho_{2}|q_{1},q_{2}\right)  |_{_{_{\gamma
_{1}}}}}\tag{A22}\label{A22}\\
&  \leq\frac{\max\{q_{1},q_{2}\}}{\mathbb{P}_{\Pi_{\vec{n}+}^{(+)}}%
^{success}\left(  \rho_{1},\rho_{2}|q_{1},q_{2}\right)  |_{_{_{\gamma_{1}}}}%
},\nonumber
\end{align}
while the expression (\ref{A20}) is true for all $\gamma_{2}>\gamma_{2}%
^{(1)}.$

On the other hand, from (\ref{A19}) it follows that $\mathbb{P}_{A|\leadsto
1\leadsto2}^{opt.success}\left(  \rho_{1},\rho_{2}|q_{1},q_{2}\right)
|_{_{\gamma_{_{1}},\gamma_{_{2}}}}$ equals to
\begin{equation}
\left(  1-\frac{\gamma_{2}}{2}\right)  \mathbb{P}_{\Pi_{\vec{n}_{+}}^{(+)}%
}^{success}\left(  \rho_{1},\rho_{2}|q_{1},q_{2}\right)  |_{_{_{\gamma_{1}}}}
\tag{A23}\label{A23}%
\end{equation}
for all $\gamma_{2}\leq\gamma_{2}^{(2)}$ where
\begin{equation}
1-\frac{\gamma_{2}^{(2)}}{2}=\frac{\max\{q_{1},q_{2}\}}{\mathbb{P}_{\Pi
_{\vec{n}_{+}}^{(+)}}^{success}\left(  \rho_{1},\rho_{2}|q_{1},q_{2}\right)
|_{_{_{\gamma_{1}}}}}, \tag{A24}\label{A24}%
\end{equation}
and is equal to $\max\{q_{1},q_{2}\}$ for all $\gamma_{2}>\gamma_{2}^{(2)}.$

Comparing (\ref{A22}) and (\ref{A24}), we conclude%
\begin{equation}
\left(  1-\frac{\gamma_{2}^{(1)}}{2}\right)  \leq\left(  1-\frac{\gamma
_{2}^{(2)}}{2}\right)  \Leftrightarrow\gamma_{2}^{(1)}\geq\gamma_{2}^{(2)},
\tag{A25}\label{A25}%
\end{equation}
so that besides $\gamma_{2}>\gamma_{2}^{(1)}$ the equality (\ref{A20}%
$\mathrm{)}$ is also true for all $\gamma_{2}^{(2)}\leq\gamma_{2}\leq
\gamma_{2}^{(1)}$.

Taking this into account for analyzing the validity of expressions (\ref{A19})
and (\ref{A20}), we come to the statement of Theorem 4.

\end{document}